\documentclass[%
 reprint,
 superscriptaddress,
 amsmath,amssymb,
 aps,
 prl,
 nobibnotes,
 longbibliography,
 floatfix,
]{revtex4-2}
\usepackage[dvipsnames]{xcolor}
\usepackage{float}
\usepackage{subcaption}
\usepackage{stackengine}
\usepackage{graphicx}
\usepackage{array}[=2016-10-06]
\usepackage{dcolumn}

\usepackage{hyperref}
\hypersetup{
  colorlinks=true,
  hypertexnames=false,
  linktocpage,
  colorlinks=true,
  urlcolor=magenta!90!black,
  linkcolor=blue!60!black,
  citecolor=black!60
}

\usepackage{amsthm,bm}
\usepackage{booktabs}
\usepackage{braket}
\usepackage{physics}
\usepackage{tikz}
\usetikzlibrary{quantikz2}

\usepackage[ruled, vlined, linesnumbered]{algorithm2e}
\usepackage{xprintlen}
\usepackage{layouts}
\usepackage{comment}
\usepackage[capitalize]{cleveref}

\usepackage{tcolorbox}

\theoremstyle{definition}
\newtheorem{theorem}{Theorem}

\newtheorem{lemma}[theorem]{Lemma}
\newtheorem{proposition}[theorem]{Proposition}

\newtheorem*{syndromerestatement}{Theorem~\ref*{proposition: syndrome sampling}}

\usepackage{mathrsfs}
\usepackage{bibunits}
\usepackage{etoolbox}
\defaultbibliographystyle{apsrev4-2}
\makeatletter
\newcommand{\bibliographynamespace}[1]{%
  \def\@extra@b@citeb{#1}%
  \def\@extra@binfo{#1}%
  \@startbibunitorrelax
  \immediate\write\@bibunitaux{\string\citation{paperBibliographyControl}}%
}
\patchcmd{\NAT@bibsetnum}{\ref{LastBibItem}}%
  {\ref{LastBibItem-\@extra@b@citeb}}{}%
  {\PackageError{paper}{Could not scope the bibliography width label}{}}
\patchcmd{\endthebibliography}{\label{LastBibItem}}%
  {\label{LastBibItem-\@extra@b@citeb}}{}%
  {\PackageError{paper}{Could not scope the final bibliography label}{}}
\makeatother
\begin{document}
\newcommand{\papertitle}{Efficient fidelity simulation of high-rate magic distillation circuits}
\title{\papertitle}
\author{Xiao Xiao}
\email{xiaoxiao@umd.edu}
\affiliation{QuEra Computing Inc., 1380 Soldiers Field Road, Boston, MA 02135, USA}
 \affiliation{
 Joint Center for Quantum Information and Computer Science, University of Maryland,
College Park, Maryland 20742, USA.
}
\author{Dominik Hangleiter}
 \affiliation{Simons Institute for the Theory of Computing, University of California at Berkeley, Berkeley, California 94270, USA.}
\affiliation{Institute for Theoretical Physics, ETH Z\"urich, 8093 Zürich, Switzerland.}
\author{J. Pablo Bonilla Ataides}
\affiliation{
 Department of Physics, Harvard University, Cambridge, Massachusetts 02138, USA.
}
\author{Rohan Mehta}
\affiliation{
 Department of Physics, Harvard University, Cambridge, Massachusetts 02138, USA.
}
\author{Varun Menon}
\affiliation{
 Department of Physics, Harvard University, Cambridge, Massachusetts 02138, USA.
}
\author{Mikhail D. Lukin}
\affiliation{
 Department of Physics, Harvard University, Cambridge, Massachusetts 02138, USA.
}
\author{Michael J. Gullans}
\email{mgullans@quera.com}
\affiliation{QuEra Computing Inc., 1380 Soldiers Field Road, Boston, MA 02135, USA}
 \affiliation{
 Joint Center for Quantum Information and Computer Science, University of Maryland,
College Park, Maryland 20742, USA.
}
\begin{bibunit}
\bibliographynamespace{main}
\makeatletter
\let\addcontentsline\@gobblethree
\makeatother
\date{\today}
\begin{abstract}
Non-Clifford gates are essential for avoiding classical simulability under the Gottesman-Knill theorem. An error-corrected circuit with non-Clifford gates will therefore be an important element for realizing useful fault-tolerant quantum algorithms. However, accurately understanding the fault-tolerant performance of such circuits faces challenges due to the classical hardness of simulating these circuits. We avoid simulating logical measurements and develop exact and efficient classical methods for benchmarking circuits
composed of $X$, CNOT, and diagonal gates in the Clifford hierarchy, including state and gate teleportation gadgets with feedforward. For circuits with third-level diagonal gates, our benchmarking algorithm efficiently samples syndrome statistics and estimates
the output logical fidelity and expectation values of certain
Clifford observables, including all Pauli observables.
The evaluation of syndrome expectation values
also extends to circuits with fourth-level diagonal gates. Our algorithm scales polynomially in the number of physical qubits $n$ and circuit depth $T$, independent of the number of logical qubits. The benchmarking applications include IQP sampling and magic state preparation/cultivation. We demonstrate the practical utility of our simulation algorithms by using them to optimize high-rate magic state distillation factories based on three copies of a [[27,3,3]] tricycle code.
\end{abstract}

\maketitle

\textit{Introduction}.---
Fault-tolerant quantum computation must suppress physical errors, while computations beyond the scope of efficient stabilizer simulation additionally require logical operations outside the Clifford group
\cite{gottesman1998heisenberg,aaronson2004improved,gidney2021stim}. Logical non-Clifford operations can be supplied by gate teleportation \cite{gottesman1999quantum} using encoded magic states prepared through distillation \cite{bravyi2005universal}, cultivation \cite{gidney2024magic,rosenfeld2025magic}, or protocols that access transversal non-Clifford gates through code switching \cite{daguerre2025experimental}. Alternatively, when the required non-Clifford operations are transversal in the chosen code, they may be applied directly to the encoded data \cite{bluvstein2024logical,hangleiter2025fault}.

As such operations are implemented on increasingly large encoded
systems, accurately predicting their performance under realistic noise becomes an essential part of designing and validating fault-tolerant protocols. For Clifford circuits subject to stochastic Pauli noise, Monte Carlo
benchmarking is efficient: one samples a fault trajectory, propagates
the faults with a stabilizer simulator, and determines the resulting
syndrome and logical action. Non-Clifford gates instead propagate Pauli
faults to non-Pauli operators, and generic exact simulation typically
incurs a cost exponential in the number of non-Clifford gates or in an
associated measure of magic \cite{surti2026efficient}. Fault-tolerant
benchmarking imposes an additional requirement: the simulator must
retain syndrome-resolved information, because the recovered logical
state depends on the decoder and on any syndrome-dependent
postselection. This makes direct benchmarking difficult for protocols
containing an extensive number of logical non-Clifford gates.

Here, we develop techniques for predicting the performance of a family of noisy, fault-tolerant non-Clifford circuits. 
%
The circuit family we consider arises naturally in quantum codes admitting transversal or constant-depth non-Clifford gates from the third level of the Clifford hierarchy. 
Examples include triorthogonal and three-dimensional color codes with transversal $T$-type gates
\cite{bravyi2012magic,kubica2015universal}, three-dimensional surface codes with transversal $\mathrm{CCZ}$ gates \cite{vasmer2018three}, and more recent quantum-LDPC constructions based on rainbow codes, algebraic products, sheaves, and homological or balanced products, whose physical non-Clifford circuits are transversal $T$ gates or constant-depth $\mathrm{CCZ}$ circuits
\cite{scruby2024quantum,golowich2025quantum,lin2024transversal,
zhu2025topological,zhu2025transversal,breuckmann2026cups,menon2026magic}.
These gates support two broad classes of applications relevant to our simulation method. First, they appear in nonuniversal but classically hard sampling circuits, including IQP-type protocols \cite{hangleiter2025fault,bluvstein2024logical}. Second, they provide building blocks for magic-state preparation. This includes distillation protocols based on triorthogonal and related codes, direct preparation of encoded magic states using transversal or constant-depth non-Clifford gates \cite{daguerre2025experimental,menon2026magic, li2025transversal}, and up to a Clifford basis change, cultivation protocols based on fault-tolerant logical checks and postselection \cite{gidney2024magic,rosenfeld2025magic,surti2026efficient}. 

\textit{Main results}.---
We develop Monte Carlo methods for sampling syndrome information and benchmarking the logical performance of 
circuits composed of $X$, CNOT, and diagonal gates in the Clifford hierarchy under circuit-level stochastic Pauli noise. The central observation is that benchmarking a fault-tolerant non-Clifford circuit is easier than fully simulating it. Our algorithms backpropagate faults, syndrome measurements, recovery operations, and code-space projectors through the ideal
circuit to the input. We exploit the fact that after propagation, Pauli stabilizer
operators, errors, and corrections lie at a lower level of the
Clifford hierarchy than the circuit.
The remaining ideal evolution then cancels from the logical
fidelity overlap, decoupling its evaluation from the explicit
non-Clifford evolution.

For circuits with third-level diagonal gates including gate teleportation, we provide efficient classical algorithms for
\begin{enumerate}
    \item exactly sampling from the syndrome distribution, 
    \item estimating the logical
    fidelity, and 
    \item efficiently evaluating Pauli-observable expectation values,
    both directly on the noisy output and after complete syndrome
    projection, Pauli recovery, or postselection.
\end{enumerate}
Furthermore, for circuits with fourth-level diagonal gates, we construct
unbiased estimators of individual stabilizer-check expectation values.
For each sampled fault trajectory, such an expectation reduces to an overlap between stabilizer states generated by a Clifford commutator.
Typical syndrome extraction circuits at any point of the circuit are included in the family, and moreover, the efficiency of our algorithms remains unchanged by adding arbitrary Clifford circuits before and after the affine-diagonal gates.
Because the level-three algorithm produces exact syndrome samples, it can be
paired with arbitrary decoders and syndrome-dependent postselection rules.

When each syndrome check is affected by faults from only a constant
number of circuit layers, as in repeated syndrome extraction
using consecutive-round parity checks, each Monte Carlo
sample takes $\mathcal O(n^3T)$ time after preprocessing. Here $n$ is the total number of simulated
data and ancilla qubits and $T$ is the circuit depth
(see \cref{subsection: computational cost}).
We stress that our simulation algorithms have no separate dependence on the number of encoded qubits, in contrast to earlier work \cite{surti2026efficient}, 
making them applicable to high-rate codes.

\textit{Affine-Diagonal circuits}.---
We consider circuits generated by $X$, $\mathrm{CNOT}$, and diagonal gates in the $\ell$th-level of the Clifford hierarchy. This gate set generates the same group as the CNOT-dihedral gate set that includes the only non-Clifford element $R_\ell=\operatorname{diag}(1,e^{2\pi i/2^\ell})$ \cite{cross2016scalable,amy2016finite}. Let $\mathcal{D}^{(\ell)}_n$ denote the abelian group of computational-basis diagonal gates in the $\ell$th level of the Clifford hierarchy~\cite{zeng2008semi,cui2017diagonal}. A finite-precision diagonal gate can be written as $D_f\ket{\mathbf{x}}=\exp[2\pi i f(\mathbf{x})/2^q]\ket{\mathbf{x}}$, where $f:\mathbb Z_2^n\rightarrow\mathbb Z_{2^q}$ is called a phase polynomial. For a monomial $g$ of degree $\deg(g)$ and phase precision $q$, its Clifford-hierarchy level is $(q-1)+\deg(g)$. Thus, for example, $\mathrm{T}$, controlled-$S$, and $\mathrm{CCZ}$ all belong to $\mathcal D_n^{(3)}$.

The gates $\mathrm{X}$ and $\mathrm{CNOT}$ generate the affine group $\mathcal A_n\cong\mathbb Z_2^n\rtimes\mathrm{GL}(n,2)$, whose elements act as $A\ket{\mathbf x}=\ket{L\mathbf x+\mathbf b}$ for $L\in\mathrm{GL}(n,2)$ and $\mathbf b\in\mathbb Z_2^n$. Since an affine transformation only permutes computational-basis states, $\mathcal A_n$ normalizes $\mathcal D_n^{(\ell)}$. The circuits generated by $\{X,\mathrm{CNOT},\mathcal D_n^{(\ell)}\}$ therefore form the affine--diagonal group
\begin{equation}
    \mathcal{AD}^{(\ell)}_n
    =\mathcal D_n^{(\ell)}\rtimes\mathcal A_n,
    \qquad
    U=D_fA .
    \label{equation: affine diagonal normal form}
\end{equation}
Here $D_f\in\mathcal D_n^{(\ell)}$ and $A\in\mathcal A_n$, and we refer to $U=D_fA$ as the \emph{affine--diagonal normal form}. It separates the diagonal phase operation $D_f$
from the affine transformation $A$ of computational-basis states. This structure is the starting point for our analysis of Pauli-error propagation and the simulation algorithms below.

\textit{Error propagation and commutators}.---
The affine--diagonal normal form directly constrains how Pauli errors propagate through the circuit. For $D\in\mathcal D_n^{(\ell)}$ and $P\in\mathcal P_n$, propagation through the diagonal component can be written as $DPD^\dagger=[D,P]P$, and the group commutator $[D,P]=DPD^{\dagger}P^{\dagger}$ is diagonal and belongs to $\mathcal D_n^{(\ell-1)}$~\cite{cui2017diagonal}. This hierarchy reduction gives the following structure for errors propagated through an affine--diagonal circuit.

\begin{lemma}[Pauli error propagation]
\label[lemma]{proposition: propagation}
Let the level $\ell\geq2$, $U\in\mathcal{AD}_n^{(\ell)}$, and
$P\in\mathcal P_n$. Then
\begin{equation}
    UPU^\dagger\in
    \mathcal{XD}_n^{(\ell-1)}
    \equiv
    \mathcal D_n^{(\ell-1)}
    \rtimes\langle X\rangle_n .
    \label{equation: propagated error group}
\end{equation}
\end{lemma}

The proof, together with the group-theoretic details of the level reduction, is given in the Supplemental Material~\cite{supplemental}. A further reduction occurs when comparing two propagated errors.

\begin{lemma}[Commutators of propagated errors]
\label[lemma]{proposition: propagated error commutator}
Let $E_1,E_2\in\mathcal{XD}_n^{(\ell-1)}$ be two propagated Pauli errors of a level-$\ell$ affine--diagonal circuit. Then
\begin{equation}
    [E_1,E_2]\in\mathcal D_n^{(\ell-2)} .
    \label{equation: propagated error commutator}
\end{equation}
\end{lemma}

For the case considered in our simulation algorithm, $\ell=3$, the propagated errors are Clifford operators in $\mathcal{XD}_n^{(2)}$, while their pairwise commutators are $Z$-type Pauli operators in $\mathcal D_n^{(1)}$. This latter property is the key ingredient used below to efficiently sample the syndrome distribution.

\textit{Syndrome sampling}.---
For Clifford circuits with Pauli noise, stabilizer measurement outcomes are determined efficiently by propagating Pauli errors. In an affine--diagonal circuit, however, a Pauli fault generally propagates to a non-Pauli operator, so the syndrome distribution can depend on the logical state. Nevertheless, the restricted error structure of $\mathcal{AD}^{(3)}$ circuits makes this distribution particularly simple. For a fixed circuit-level Pauli error configuration, let $e^{\mathrm{ini}}$ denote the product of all errors propagated to the beginning of the circuit, and let $M^{\mathrm{ini}}=U^\dagger M U$ be a final stabilizer measurement propagated to the same point. Its expectation value is
\begin{align}
    \langle M\rangle
    =
    \Tr\!\left(M^{\mathrm{ini}}\,
    e^{\mathrm{ini}}\rho_{\mathrm{ini}}(e^{\mathrm{ini}})^\dagger\right)
    =
    \Tr\!\left(
    [(e^{\mathrm{ini}})^\dagger,M^{\mathrm{ini}}]\,
    \rho_{\mathrm{ini}}
    \right).
    \label{equation: syndrome expectation}
\end{align}
For $\mathcal{AD}^{(3)}$ circuits, the commutator is a $Z$ type Pauli string by \cref{proposition: propagated error commutator}. Hence, for a stabilizer input state, $\langle M\rangle\in\{0,\pm1\}$. The same expression remains efficiently evaluable for $\mathcal{AD}^{(4)}$ circuits, where the commutator is a Clifford circuit.

\begin{theorem}[Syndrome sampling]
\label{proposition: syndrome sampling}
Consider a stabilizer input state evolved by an
$\mathcal{AD}^{(\ell)}$ circuit and followed by commuting stabilizer
measurements. For a fixed configuration of circuit-level Pauli faults,
stabilizer expectation values can be evaluated efficiently for
$\ell\leq4$. For $\ell\leq3$, the complete syndrome can moreover be
sampled exactly in polynomial time using \cref{algorithm: syndrome sampling concise}.
\end{theorem}

Let the measured stabilizer group $\mathcal M=\langle M_1,\ldots,M_m\rangle$ and write
$M^{(\mathbf r)}=\prod_iM_i^{r_i}$ for $\mathbf r\in\mathbb Z_2^m$.
Given a fault configuration $e$, the labels $\mathbf{r}$ for which the syndrome is deterministic, i.e., $\langle M^{(\mathbf r)}\rangle=\pm1$ form a
linear subspace $V\subseteq\mathbb Z_2^m$. If $\mathbf s_0$ satisfies
$(-1)^{\mathbf r\cdot\mathbf s_0}=\langle M^{(\mathbf r)}\rangle$
for all $\mathbf r\in V$, the syndrome distribution is
\begin{equation}
    \Pr(\mathbf s|e)
    =
    \frac{1}{2^{m-\dim V}}\,
    \mathbf 1_{V^\perp}(\mathbf s+\mathbf s_0),
    \label{equation: affine syndrome distribution}
\end{equation}
where $\mathbf 1_{V^\perp}$ denotes the indicator of $V^\perp$.
Thus, for every sampled fault configuration, the syndrome is uniform on
the affine space $\mathbf s_0+V^\perp$. The subspace and shift
are obtained by binary Gaussian elimination, giving the
sampling procedure in \cref{algorithm: syndrome sampling concise}. The proof of the affine-subspace form and further implementation details are given in the Supplemental Material \cref{subsection: syndrome sampling}.

\begin{algorithm}[t]
\caption{Syndrome sampling for $\mathcal{AD}^{(3)}$ circuits}
\label{algorithm: syndrome sampling concise}
\KwIn{Fault trajectory $e$, measurements $\{M_i\}$, initial stabilizers $\{S_j\}$.}

Propagate $e$ and each $M_i$ to the input, obtaining
$e^{\mathrm{ini}}$ and $M_i^{\mathrm{ini}}$.

Set $K_i=[(e^{\mathrm{ini}})^\dagger,M_i^{\mathrm{ini}}]$ and
$B_{ij}=\langle K_i,S_j\rangle_{\mathrm{sp}}$.

Find a basis matrix $G$ of $V=\ker B^T$ by Gaussian elimination.

For each row $\mathbf v_i$ of $G$, compute $\mathbf{a}$ where
$(-1)^{a_i}=\langle M^{(\mathbf v_i)}\rangle$.

Solve $G\mathbf s_0=\mathbf a$ and find a basis $N$ of $\ker G$.

Sample $\mathbf z$ uniformly and return
$\mathbf s=\mathbf s_0+N\mathbf z$.
\end{algorithm}

\textit{Logical fidelity}.---
Together with syndrome sampling, the Clifford structure of
propagated errors allows exact and efficient estimation of the output
fidelity after Pauli recovery or postselection. We include the data and all stabilizer state ancillas in a
single system and let $\Pi_{0}^{\mathrm{in}}$ and $\Pi_{0}^{\mathrm{out}}$ denote the
stabilizer projectors defining the ideal input and output subspaces. 
Neither the data code nor the ancilla stabilizer state needs to be identical
at the two boundaries. We require only that the ideal circuit $U$ maps one
reference subspace to the other,
\begin{equation}
    U\Pi_{0}^{\mathrm{in}}U^\dagger=\Pi_{0}^{\mathrm{out}} .
    \label{equation: projector intertwining}
\end{equation}
This condition includes logical gates, code deformations, and ideal syndrome
extraction circuits.

Recovery or postselection may use incomplete syndrome information. For example, a transversal final measurement may reveal only one type of CSS syndrome, or one may postselect on only a subset of stabilizer checks. Let $\mathcal M$ denote the measured stabilizer subgroup and
$\Pi_{\mathbf{s},\mathcal M}^{\mathrm{out}}$ the projector onto its syndrome $\mathbf{s}$
sector. A Pauli recovery $P_{\mathbf{s}}$ that maps this sector to the $+1$ sector satisfies $P_{\mathbf{s}}\Pi_{\mathbf{s},\mathcal M}^{\mathrm{out}}
=\Pi_{0,\mathcal M}^{\mathrm{out}}P_{\mathbf{s}}$. For fidelity evaluation, we supplement the measured syndrome
with virtual measurements of a complete set of stabilizers
of the ideal ancilla state. We include these outcomes in
$\mathbf{s}$ and the corresponding ancilla recovery in
$P_{\mathbf{s}}$, thereby restoring the ancillas to their
ideal state.

For a sampled circuit-level Pauli fault configuration $e$, let $C_e$ be the
corresponding Clifford error propagated to the input and define $C_{\mathbf{s}}=U^{\dagger}P_{\mathbf{s}}U$.  The joint-state overlap then gives the output fidelity of
the data blocks after the chosen recovery or postselection:
\begin{equation}
\begin{split}
    F_{\mathrm{out}}
    &=
    \sum_{e,\mathbf{s}} p_e
    \left|
    \bra{\overline{\phi}}
    P_{\mathbf{s}}\Pi_{\mathbf{s},\mathcal M}^{\mathrm{out}}
    UC_e
    \ket{\overline{\psi}}
    \right|^2                                      \\
    &=
    \sum_{e,\mathbf{s}} p_e
    \left|
    \bra{\overline{\psi}}
    C_{\mathbf{s}}C_e
    \ket{\overline{\psi}}
    \right|^2 .
\end{split}
\label{equation: output fidelity}
\end{equation}
Since the
ideal encoded data plus ancilla output $\ket{\overline{\phi}}=U\ket{\overline{\psi}}$
is stabilized by $\mathcal{M}$, the projector $\Pi_{0,\mathcal M}^{\mathrm{out}}$ is absorbed directly into the
target bra vector.
For $U\in\mathcal{AD}^{(3)}$, both $C_{\mathbf{s}}$ and $C_e$ are
Clifford operators, so every term is an efficiently computable stabilizer-state
overlap.  The postselection is handled by rejecting the branch from the syndrome sampling procedure.

To obtain a meaningful notion of logical fidelity, any syndrome unavailable from the
implemented circuit can be supplied by perfect virtual syndrome
measurements, followed by Pauli recovery to the code
space~\cite{hangleiter2025fault}.

Note that the method extends to circuits with feedforward logical operations such as state and gate teleportation by rewriting the original circuit through a unitary representation. See \cref{sec:teleportation} for details.

\textit{Observable expectation}.---
The same structure also allows observables to be evaluated after syndrome
projection and recovery. Let $\Pi_{\mathbf s}^{\mathrm{out}}$ be a complete output syndrome projector,
including the data and ancilla stabilizers, and let $Q_{\mathbf s}$ be a Pauli
correction satisfying
$Q_{\mathbf s}\Pi_{\mathbf s}^{\mathrm{out}}Q_{\mathbf s}^\dagger
=\Pi_{0}^{\mathrm{out}}$. Let $\ket{\Omega}$ be the initial stabilizer state.
For a fault trajectory $e$ with probability $p_e$, let $C_e$ denote the
corresponding Clifford error propagated to the input. Defining
$C_{\mathbf s}=U^\dagger Q_{\mathbf s}U$ and
$O_{\mathrm{in}}=U^\dagger O U$, the expectation value after correction,
averaged over fault trajectories and syndrome outcomes, is
\begin{equation}
    \langle O\rangle_{\mathrm{out}}
    =\sum_{e,\mathbf s}p_e
    \langle\Omega|
    C_e^\dagger C_{\mathbf s}^\dagger
    \Pi_{0}^{\mathrm{in}} O_{\mathrm{in}}\Pi_{0}^{\mathrm{in}}
    C_{\mathbf s}C_e
    |\Omega\rangle,
    \label{equation: corrected observable}
\end{equation}
if $U\in\mathcal{AD}^{(3)}_n$ satisfies
\cref{equation: projector intertwining}.
For $O\in\mathcal{XD}^{(2)}_n$, the operators
$C_{\mathbf s}$, $C_e$, and $O_{\mathrm{in}}$ are Clifford,
and the normalized projection
$\Pi_0^{\mathrm{in}}C_{\mathbf s}C_e\ket{\Omega}$
is a stabilizer state. Its conditional expectation of
$O_{\mathrm{in}}$ is therefore efficiently computable.
Averaging over sampled fault trajectories and complete
syndromes gives an unbiased estimator with polynomial
cost per sample. For postselection, we average over accepted
samples. See \cref{proposition: observable expectation}
in the Supplemental Material for the derivation.

\begin{figure}[t]
  \centering
  \includegraphics[width=\columnwidth]{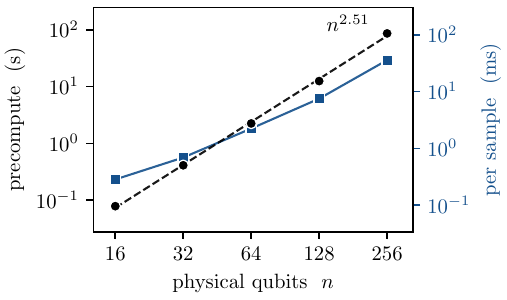}
\caption{Cost of the fidelity estimation pipeline on one CPU versus the number of
    physical qubits $n$, for the hypercube-IQP circuits using the $[[8,3,2]]$ code \cite{hangleiter2025fault}.
    Black circles: one-time precomputation (simulator construction, back-propagation
    of faults and measured stabilizers, and the fault--measurement commutator
    tables), with a fitted power law $n^{2.51}$ (dashed). Blue squares: cost per
    Monte-Carlo sample, obtained by subtracting the precomputation from the total run
    time and dividing by the number of samples. No power law is quoted for the latter
    because its local slope still rises across this range ($1.3 \to 2.3$). See
    \cref{subsection: computational cost} for the scaling of the individual
    subroutines.}

  \label{fig:runtime}
\end{figure}

\textit{Example 1 hypercube IQP sampling}.---
A degree-$D$ IQP sampling task prepares states in the X basis, uses random gates from $\mathrm{C^{k}Z}$ with $k=1,..., D-1$, and finally measures in the X basis \cite{shepherd2008instantaneous,shepherd2010quantum,bremner2011classical}. Hypercube-IQP circuits are encoded versions of IQP circuits augmented by CNOT gates, co-designed with the hypercube code
family $[[2^D,D,2]]$~\cite{hangleiter2025fault}.
These circuits have been experimentally
demonstrated~\cite{bluvstein2024logical}. These circuits have been shown to be classically hard to sample from under a widely believed complexity-theoretic conjecture \cite{hangleiter2025fault,bremner2016average,bremner2017achieving}. A logical hIQP circuit with hypercube dimension $\mathfrak{D}$ encodes $D2^{\mathfrak{D}}$ qubits using $b=2^{\mathfrak{D}}$ blocks of the $[[2^D,D,2]]$ code. By placing code blocks on the vertices of a $\mathfrak{D}$-dimensional hypercube, each hyperlayer of a circuit with hyperdepth $\mathfrak{L}$ consists of random in-block IQP and CNOT gates, as well as transversal CNOT gates between all vertices connected by an edge of the hypercube.

In the $D=3$ example, the non-Clifford CCZ gate is encoded using eight physical $\mathrm{T}$ or $\mathrm{T}^{\dagger}$ gates transversally on the $[[8,3,2]]$ code. Similarly, logical Z and CZ gates are implemented using $\mathrm{Z}$ and $\mathrm{S/S^{\dagger}}$ physical gates. Therefore, the physical circuits are also affine-diagonal circuits with $\ell=3$ and the fidelity can be simulated efficiently. In \cref{fig:runtime}, we show the computational cost of the pre-computation and logical fidelity calculation of each sample of Pauli fault configuration, for $\mathfrak{D}=1,...,5$. In particular, we show that the runtime scales polynomially with the number of physical qubits and circuit depth. See details in \cref{subsection: computational cost}.

\begin{figure}[t]
  \centering
  \includegraphics[width=\columnwidth]{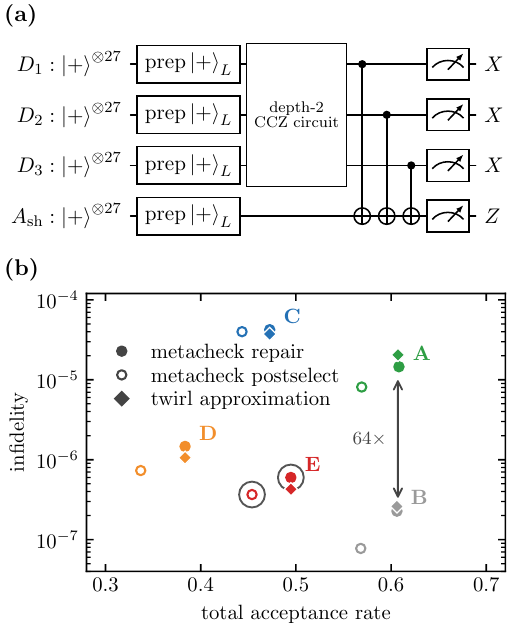}
  \caption{(a) Protocol E: three [[27,3,3]] data blocks are prepared in $|+\rangle_L$, a depth-2 transversal CCZ circuit builds the encoded hypergraph state, and one shared ancilla block $A_{\rm sh}$ extracts the post-CCZ $Z$ syndrome. (b) Infidelity of the 9-qubit hypergraph state versus total acceptance rate for five treatments of the Z syndrome — A: none. B: noiseless (lower bound). C: bare ancillas. D: one Steane block per data block. E: one shared block, circled, as drawn in (a) — with the prep metachecks repaired (filled circle) or postselected (hollow). The diamond data points correspond to results with the approximation using twirling \cite{menon2026magic}. The double arrow marks the A–B gap. Simulations were performed with two-qubit depolarizing noise at a total error rate of $10^{-3}$ and CCZ gate noise at a total error rate of $2\times10^{-3}$ with a $10{:}1$ bias toward $Z$-type errors, following the noise model of Ref.~\cite{menon2026magic}.}
  \label{fig:tricycle}
\end{figure}

\textit{Example 2: Tricycle-code magic generation}.---
Tricycle codes are CSS codes defined by block parity-check matrices constructed from sums of permutation matrices, extending the group-algebra structure of bivariate bicycle codes to three dimensions
\cite{panteleev2021degenerate,bravyi2024high}. Formally, they arise from three-dimensional balanced products of classical linear codes over Abelian group algebras. Suitable instances admit constant-depth physical $\mathrm{CCZ}$ circuits that preserve the code space and implement networks of
logical $\mathrm{CCZ}$ gates
\cite{lin2024transversal,menon2026magic,li2025transversal}. Such circuits, when combined with syndrome-based postselection and single-shot property \cite{bombin2015single} in the appropriate basis, make tricycle codes promising for fault-tolerant preparation of $\mathrm{CCZ}$-type
magic states. Existing numerical studies rely on simplified treatments of syndrome sampling and the correlated errors propagated by the physical $\mathrm{CCZ}$
circuit, motivating a direct circuit-level simulation of the full non-Clifford protocol \cite{menon2026magic,li2025transversal}.

Here we explore different syndrome extraction and postselection strategies and simulate the logical fidelity of the final logical state. All methods start from physical product states and perform a single round of syndrome extraction, leveraging the single-shot $\ket{\overline{+}}$ state preparation property. We perform postselection on the X syndrome and perform either correction or postselection on the Z syndrome based on the Z meta-checks. Pauli X correction is then applied. This is followed by finite-depth CCZ gates. Finally, we postselect on trivial $X$ syndrome and compare
different strategies for additional $Z$-syndrome extraction: ideal extraction, bare-syndrome extraction, Steane-style extraction \cite{steane1997active}, and a shared-block extraction. We add an additional fictitious $ Z$ syndrome extraction at the end and correct the postselected state back to the code space to calculate logical fidelity. For all decoding, we use the BP+OSD decoder \cite{roffe2020decoding}. The additional $Z$ checks address fault-propagation mechanisms
identified in Ref.~\cite{menon2026factories}, targeting faults that remain undetected by $X$-syndrome postselection, as analyzed in
\cref{sec:ccz-single-faults}. In \cref{fig:tricycle}, we show the logical fidelities and acceptance fractions for various Z syndrome extraction methods, and compare our exact simulation with an approximation
that replaces CZ errors arising from fault propagation
through the CCZ circuit by their Pauli-twirled
channels~\cite{menon2026magic}, applied to the initial
nine-logical-qubit state $\ket{\overline{+}}^{\otimes 9}$. We observe no large deviation between the exact and approximate results here. See \cref{sec:simulation-details} for the details of the simulation and \cref{sec:syndrome-protocols} for each protocol (A-E).

\textit{Example 3: Cultivating magic states}.---
An alternative approach to improving magic-state fidelity is to measure the relevant logical Clifford observable and postselect on the desired eigenvalue.
This approach underlies magic-state cultivation and related verification
schemes. The logical magic state $\ket{\overline{T}}=\overline{T}\ket{\overline{+}}$ can be cultivated by measuring the logical checks $\overline{H}_{XY}=(\overline{X}+\overline{Y})/\sqrt2$ using CNOTs and third-level diagonal gates.
Examples directly compatible with our simulator include the double-checking circuit
in color codes~\cite{gidney2024cultivation}, the regular and $\mathbb{RP}^2$ surface codes~\cite{sahay2026foldtransversal,claes2025twoqubit,vaknin2026highrate,chen2026rp2}, and similar phase kickback verification subroutines
adopted in Refs.~\cite{lee2025colorcodes,hirano2025latticesurgery}.

Cultivation of $\ket{\overline{H}}$ state in the Steane
code~\cite{itogawa2025zerolevel,hirano2024leveraging,surti2026efficient}
instead measures $\overline H$ using physical controlled-Hadamard gates.
Define
\begin{equation}
    K=SH=\frac{1}{\sqrt2}
    \begin{pmatrix}1&1\\ i&-i\end{pmatrix}.
\end{equation}
For ancilla control $a$ and data target $d$, the Clifford frame change on data qubits gives
$(I\otimes K^\dagger)\,\mathrm{C}H_{a,d}\,(I\otimes K)
=T_a^\dagger\mathrm{C}S_{a,d}\mathrm{CNOT}_{a\to d}
\in\mathcal{AD}^{(3)}$.
The transformed ideal preparation is
$\overline T^\dagger\ket{\overline+}$, which can be prepared by applying
a third-level diagonal gate to the stabilizer state $\ket{\overline+}$.
Pauli errors stay as Pauli errors in the new frame.
Together with deferred readout and virtual ancilla recovery, these
representations permit our syndrome-sampling and fidelity-estimation
methods under circuit-level Pauli noise. An interesting direction is to extend this analysis to
complete magic-state preparation
protocols~\cite{gidney2024cultivation,hetenyi2026gauging,
rosenfeld2025experiment}.

\textit{Discussion and Outlook}.---
We have shown that affine--diagonal circuits permit exact and efficient
classical benchmarking of a broad class of fault-tolerant non-Clifford
circuits under circuit-level stochastic Pauli noise. We find that simulating all measurements is not necessary other than the syndrome information, even in the presence of logical feedforward operations. This demonstrates that an extensive number of non-Clifford gates does not necessarily preclude efficient fault-tolerant simulation for benchmarking purposes, when their error propagation possesses sufficient algebraic structure. Natural extensions include incorporating noise processes beyond Pauli
channels, such as atom loss and coherent errors, and benchmarking
magic-state factories based on higher-distance qLDPC codes. Finally, extending our fidelity estimators to affine--diagonal blocks interleaved with Clifford operations could provide a route to benchmarking circuits capable of universal quantum computation.

\emph{Note added:} While finalizing this manuscript, we became aware of independent works developing efficient methods for benchmarking
noisy fault-tolerant non-Clifford blocks and magic-state preparation
protocols~\cite{de2026efficiently,takada2026exact}.
Our approach relies on backpropagating errors through affine--diagonal
circuits and also provides logical-observable estimation and,
separately, stabilizer-check expectation estimators for fourth-level
diagonal gates.

\begin{acknowledgments}
We thank Daniel Gottesman, Yijia Xu, Frank Zhang, Rafael Haenel, Samuel Tan, Yifan Hong, Christine Li, Jonathan Wurtz for helpful discussions. This research was supported in part by NSF QLCI grant OMA-2120757, NSF QLCI grant OMA-2553574, NSF NQVL:Design:ORAQL grant 2533041 and ARO grant W911NF-23-1-0258.
DH was supported by a Simons postdoctoral fellowship through DOE QSA and NSF QLCI Grant No.\ 2016245, and by the Swiss National Science Foundation through Ambizione Grant No.\ 223764.
Large language models were used to assist with numerical simulations and to improve the clarity and grammar of the manuscript.
\end{acknowledgments}

\nocite{aaronson2004improved,gidney2021stim,gottesman1999quantum,zeng2008semi,cross2016scalable,amy2016finite,cui2017diagonal,hangleiter2025fault}

\putbib[references]
\end{bibunit}

\clearpage
\onecolumngrid
\begin{bibunit}
\bibliographynamespace{supplement}
\begin{center}
    \textbf{\large Supplemental Material for\\
    \papertitle}

    \vspace{0.6em}

    Xiao Xiao$^{1,2}$, Dominik Hangleiter$^{3,4}$, J. Pablo Bonilla Ataides$^{5}$,
    Rohan Mehta$^{5}$, Varun Menon$^{5}$, Mikhail D. Lukin$^{5}$,
    and Michael J. Gullans$^{1,2}$

    \vspace{0.4em}

    {\itshape
    $^{1}$QuEra Computing Inc., 1380 Soldiers Field Road, Boston, MA 02135, USA\\
    $^{2}$Joint Center for Quantum Information and Computer Science,
    University of Maryland, College Park, Maryland 20742, USA.\\
    $^{3}$Simons Institute for the Theory of Computing,
    University of California at Berkeley, Berkeley, California 94270, USA.\\
    $^{4}$Institute for Theoretical Physics, ETH Z\"urich, 8093 Z\"urich, Switzerland.\\
    $^{5}$Department of Physics, Harvard University, Cambridge, Massachusetts 02138, USA.}
\end{center}

\setcounter{section}{0}
\setcounter{equation}{0}
\setcounter{figure}{0}
\setcounter{page}{1}
\setcounter{secnumdepth}{3}                                      
\renewcommand{\thesection}{S\arabic{section}}                   
\renewcommand{\thesubsection}{\thesection.\arabic{subsection}}  
\makeatletter
\renewcommand{\p@subsection}{}
\makeatother
\renewcommand{\theequation}{S\arabic{equation}}                 
\renewcommand{\thefigure}{S\arabic{figure}}                     

\tableofcontents          

\section{Preliminary\label{section: preliminary}}

%
%
%
%

We collect the notation and background used throughout. A calligraphic
symbol denotes a group or set and the corresponding roman symbol an element
of it. Computational-basis states are labelled by bit strings
$\mathbf{x}=(x_1,\dots,x_n)\in\mathbb{Z}_2^n$, we write $[n]=\{1,\dots,n\}$, and
$\mathrm{wt}(\mathbf{a})$ denotes the Hamming weight of
$\mathbf{a}\in\mathbb{Z}_2^n$.

\subsection{Pauli group and stabilizer codes}
\label{subsection: stabilizer codes}
We define the $n$-qubit Pauli group as
$\mathcal{P}_n=\{1,i,-1,-i\}\cdot\{I, X, Y, Z\}^{\otimes n}$, generated by tensor
products of single-qubit Pauli operators. Any two elements $e,f\in\mathcal{P}_n$
either commute or anticommute, and we write the binary commutator
$\langle e,f\rangle\in\{0,1\}$ where $\langle e,f\rangle=0$ iff $e$ and $f$
commute. A stabilizer code is specified by the stabilizer, an abelian subgroup
$\mathcal{S}\leq\mathcal{P}_n$ with $-I\notin\mathcal{S}$, generated by $n-k$
independent stabilizer generators $\mathcal{S}=\langle S_1,\dots,S_{n-k}\rangle$.
Its code space, which encodes the Hilbert space of $k$ logical qubits, is the
simultaneous $+1$ eigenspace of $\mathcal{S}$, i.e.\ the image of the
zero-syndrome projector
$\Pi_{0}=\tfrac{1}{|\mathcal{S}|}\sum_{S\in\mathcal{S}}S$. A logical operation is
a unitary $\overline{U}$ that preserves the code space,
$\overline{U}\,\Pi_{0}\,\overline{U}^{\dagger}=\Pi_{0}$.

Error correction proceeds by measuring a generating set
$\{M_1,\dots,M_m\}$ of (a subgroup of) $\mathcal{S}$. In the absence of noise, the
outcomes are fixed by the code. A Pauli error $e\in\mathcal{P}_n$ flips the $i$th
outcome exactly when it anticommutes with $M_i$, producing the \emph{syndrome}
$\mathbf{s}=(s_1,\dots,s_m)\in\mathbb{Z}_2^m$ with $s_i=\langle e,M_i\rangle$. An
error is \emph{detectable} if it anticommutes with some measured generator.
A decoder maps an observed syndrome $\mathbf{s}$ to a correction
$r_{\mathbf{s}}$, and the associated recovery map $\mathcal{R}$ returns the state
to the code space. A logical failure occurs when the residual $r_{\mathbf{s}}e$
is a nontrivial logical operator.

Throughout, we model noise as circuit-level stochastic Pauli noise where we apply Pauli errors sampled from a probability
distribution. For Clifford circuits both the error propagation and the
syndrome sampling are efficient in the stabilizer/tableau
formalism~\cite{aaronson2004improved,gidney2021stim}. The central difficulty
addressed in this work is to retain efficiency in the presence of the
non-Clifford diagonal gates introduced below. As a figure of merit we use the
logical fidelity of the recovered state relative to the ideal encoded output,
\begin{equation}
    \overline{F}=\Tr\!\bigl[\mathcal{R}(\rho)\,\overline{\sigma}\bigr],
    \label{equation: logical fidelity def}
\end{equation}
where $\rho$ is the noisy output state of the encoded circuit, $\mathcal{R}$ the
recovery (syndrome measurement and correction) applied to it, and
$\overline{\sigma}=\overline{U}\ket{\overline{\psi}}\!\bra{\overline{\psi}}\overline{U}^{\dagger}$
the ideal logical output.

\subsection{Clifford hierarchy}

The Clifford hierarchy is a collection of sets of operators $\{\mathcal{C}^{(\ell)}_n\}_{\ell}$. The first level of the Clifford hierarchy is the Pauli group with an
arbitrary global phase,
$\mathcal{C}^{(1)}_n=\{\mathrm{e}^{\mathrm{i}\phi}\}\cdot\mathcal{P}_n$, and the higher
levels are defined recursively by~\cite{gottesman1999quantum}
\begin{equation}
    \mathcal{C}^{(\ell)}_n=\bigl\{\,U:\ U e U^{\dagger}\in\mathcal{C}^{(\ell-1)}_n,\
    \forall e\in\mathcal{P}_n\,\bigr\}.
    \label{equation: clifford hierarchy}
\end{equation}
The second level $\mathcal{C}^{(2)}_n$ is the Clifford group, i.e.\ the
normalizer of $\mathcal{P}_n$. It is generated by $\{H,S,\mathrm{CNOT}\}$ and maps
Pauli operators to Pauli operators. For $\ell\geq 3$ the level
$\mathcal{C}^{(\ell)}_n$ is not a group ~\cite{zeng2008semi}. Gates first
appearing at the third level include the
$T=\mathrm{diag}(1,\mathrm{e}^{\mathrm{i}\pi/4})$ gate and the
controlled--controlled-$Z$ gate $\mathrm{CCZ}$, and the Clifford group augmented
by any single gate from $\mathcal{C}^{(3)}_n$ is
universal~\cite{gottesman1999quantum}. The circuit family studied here is the generalized
CNOT--dihedral family generated by $X$, $\mathrm{CNOT}$, and diagonal hierarchy
gates~\cite{cross2016scalable,amy2016finite}, which we introduce next.

\subsection{Diagonal gates in the Clifford hierarchy}
\label{subsection: diagonal gates}
We denote by $\mathcal{D}^{(\ell)}_n$ the subset of gates in
$\mathcal{C}^{(\ell)}_n$ that are diagonal in the computational basis. Although
the full levels in the Clifford hierarchy are not groups for $\ell\geq 3$, $\mathcal{D}^{(\ell)}_n$ is an
abelian group for every $\ell\in\mathbb{N}^{+}$~\cite{zeng2008semi,cui2017diagonal}.
We adopt the conventions $\mathcal{D}^{(0)}_n=\{I_n\}$, and $\mathcal{D}^{(1)}_n$
is the group of diagonal ($Z$-type) Paulis.
A diagonal gate of \emph{precision} $q$ (phases that are integer multiples of
$2\pi/2^{q}$) is fixed by a function $f:\mathbb{Z}_2^n\to\mathbb{Z}_{2^{q}}$
through
\begin{equation}
    U=\sum_{\mathbf{x}\in\mathbb{Z}_2^n}
    \exp\!\Bigl[\tfrac{2\pi\mathrm{i}}{2^{q}}\,f(\mathbf{x})\Bigr]
    \ket{\mathbf{x}}\!\bra{\mathbf{x}}.
    \label{equation: diagonal phase function}
\end{equation}
Since $x_j^2=x_j$ for $x_j\in\{0,1\}$, the function $f$ admits a unique
representation as a \emph{phase polynomial} with coefficients
$c_{\mathbf{a}}\in\mathbb{Z}_{2^{q}}$,
\begin{equation}
    f(\mathbf{x})=\sum_{\mathbf{a}\in\mathbb{Z}_2^n}
    c_{\mathbf{a}}\prod_{j=1}^{n}x_j^{a_j} \pmod{2^{q}}.
    \label{equation: phase polynomial}
\end{equation}
Any diagonal gate of finite precision factorizes into the generating gates
$U_{q,\mathbf{a}}$, each labeled by a monomial $\prod_{j=1}^{n}x_j^{a_j}$ with
$\mathbf{a}\in\mathbb{Z}_2^n$,
\begin{equation}
    U_{q,\mathbf{a}}=\sum_{\mathbf{x}\in\mathbb{Z}_2^n}
    \exp\!\Bigl[\tfrac{2\pi\mathrm{i}}{2^{q}}\prod_{j=1}^{n}x_j^{a_j}\Bigr]
    \ket{\mathbf{x}}\!\bra{\mathbf{x}}.
    \label{equation: monomial diagonal gate}
\end{equation}
The hierarchy level of a generating gate $U_{q,\mathbf{a}}$ is fixed jointly by
its precision $q$ and degree $\mathrm{wt}(a)$. With
\begin{equation}
    \mathcal{W}_{\ell}=\bigl\{(q,\mathbf{a})\ :\ (q-1)+\mathrm{wt}(\mathbf{a})=\ell\bigr\},
    \label{equation: W_ell}
\end{equation}
the diagonal group $\mathcal{D}^{(\ell)}_n$ is generated by all such gates up to
level $\ell$ together with the global phase~\cite{cui2017diagonal},
\begin{equation}
    \mathcal{D}^{(\ell)}_n=\langle\mathrm{e}^{\mathrm{i}\phi}\rangle\cdot
    \bigl\langle\,\mathcal{U}_{\leq \ell}\,\bigr\rangle,\qquad
    \mathcal{U}_{\leq \ell}=\bigcup_{\ell'\leq \ell}
    \{\,U_{q,\mathbf{a}} : (q,\mathbf{a})\in\mathcal{W}_{\ell'}\,\}.
    \label{equation: diagonal generating set}
\end{equation}
This reproduces the standard assignments (see also \cref{fig: clifford hierarchy}):
\begin{equation}
    \renewcommand{\arraystretch}{1.3}
    \begin{array}{l@{\quad}c@{\quad}c@{\quad}c}
    \text{gate} & \text{phase }2\pi f/2^{q} & (q,\mathrm{wt}) & \text{level}\\[-0.1em]
    \hline
    Z    & \tfrac{2\pi}{2}\,x_1        & (1,1) & 1\\
    S    & \tfrac{2\pi}{4}\,x_1        & (2,1) & 2\\
    T    & \tfrac{2\pi}{8}\,x_1        & (3,1) & 3\\
    \mathrm{CZ}  & \tfrac{2\pi}{2}\,x_1x_2   & (1,2) & 2\\
    \mathrm{CS}  & \tfrac{2\pi}{4}\,x_1x_2   & (2,2) & 3\\
    \mathrm{CCZ} & \tfrac{2\pi}{2}\,x_1x_2x_3& (1,3) & 3
    \end{array}
    \nonumber
\end{equation}
so that, for example, $\mathrm{CCZ}$ and $T$ are both third level despite their
differing precisions: the level is set by precision \emph{and} degree, not by
precision alone. The circuits of interest are generated by the gate set
$\{X,\mathrm{CNOT},\mathcal{D}^{(\ell)}_n\}$, with $\ell=3$ in our concrete
algorithm. Their group structure and the resulting Pauli-error propagation are
analyzed in \cref{section: IQP propagation}.

\begin{figure}
    \centering
    \includegraphics[width=0.35\linewidth]{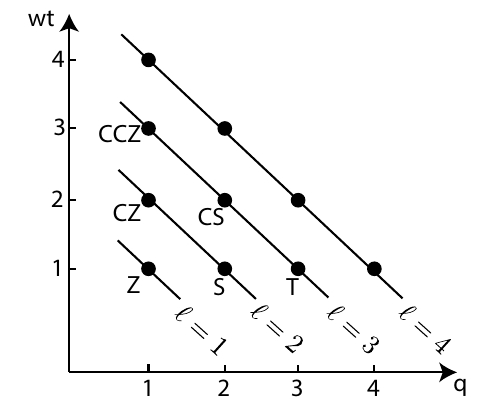}
    \caption{Generating gatesets for the group of diagonal gates in the $\ell$th level of the Clifford hierarchy. The x-axis and y-axis label the precision $q$ and weight $wt$ of the monomial of \cref{equation: monomial diagonal gate}.}
    \label{fig: clifford hierarchy}
\end{figure}

\subsection{Semidirect products and circuit normal form}
\label{subsection:semidirect-product}
A group \(G\) is said to be the internal semidirect product of a normal subgroup \(N \trianglelefteq G\) and a subgroup \(H \leq G\), denoted by \(G = N \rtimes H\), if \begin{equation} G = NH \qquad \text{and} \qquad N \cap H = \{e\}. \end{equation}
Every element \(g \in G\) then admits a unique factorization
\begin{equation}
    g = nh,
    \qquad
    n \in N,
    \quad
    h \in H.
\end{equation}
The normality of \(N\) ensures that conjugation by elements of \(H\)
preserves \(N\), namely,
\begin{equation}
    hnh^{-1} \in N
\end{equation}
for every \(n \in N\) and \(h \in H\).
When \(G\) is a group of unitary operators generated by two sets of gates,
this factorization gives a circuit normal form. In particular, any unitary
\(U \in G\) can be written as
\begin{equation}
    U = U_N U_H,
    \qquad
    U_N \in N,
    \quad
    U_H \in H,
\end{equation}
where \(U_N\) and \(U_H\) denote subcircuits generated by the respective gate
sets. Since operators act on states from right to left, this factorization
corresponds to applying the \(H\) subcircuit first, followed by the \(N\)
subcircuit:
\begin{equation}
    \ket{\psi}
    \longmapsto
    U_H \ket{\psi}
    \longmapsto
    U_N U_H \ket{\psi}.
\end{equation}
In this work, the normal subgroup \(N\) will
be the diagonal gates \(\mathcal{D}_n^{(\ell)}\). The resulting normal form
underlies the subcircuit representation used below to prove the properties of
Pauli error propagation and their commutation relations.

\section{Propagation through affine--diagonal circuits}
\label{section: IQP propagation}

The affine--diagonal circuits are the main subcircuits of interest in this work.
They are generated by the gate set $\{X,\mathrm{CNOT},\mathcal{D}^{(\ell)}_n\}$,
combining the Pauli $X$ gate, the $\mathrm{CNOT}$ gate, and the diagonal
Clifford-hierarchy gates $\mathcal{D}^{(\ell)}_n$ of
\cref{subsection: diagonal gates}. They generalize the CNOT-dihedral
circuits~\cite{cross2016scalable,amy2016finite}, which are recovered when
$\ell=3$ and the diagonal layer is restricted to the $T$ gate. The degree-$\ell$
logical IQP circuits~\cite{hangleiter2025fault} are the special case in which the
diagonal layer consists of IQP phase gates. Below we show that such a circuit
admits an affine--diagonal normal form, and that this normal form makes the
propagation of Pauli errors explicit: a level-$\ell$ circuit propagates Pauli
errors into $\mathcal{XD}^{(\ell-1)}_n$, and the commutator of two propagated
errors lies in $\mathcal{D}^{(\ell-2)}_n$.

\subsection{Affine--diagonal normal form}

Let $\langle X\rangle_n$ and $\langle\mathrm{CNOT}\rangle_n$ be the group generated by single-qubit Pauli $X$ gates and $\mathrm{CNOT}$ gates for a $n$-qubit circuit respectively, then circuit generated by $X$ and $\mathrm{CNOT}$ gates is
\begin{equation}
    \mathcal{A}_n=\langle X\rangle_n\rtimes\langle\mathrm{CNOT}\rangle_n
    \cong\mathbb{Z}_2^{\,n}\rtimes\mathrm{GL}(n,2)
    \label{equation: affine group}
\end{equation}
be the affine group, whose elements $A\in\mathcal{A}_n$ act on computational-basis
states as $A\ket{\mathbf{x}}=\ket{L\mathbf{x}+\mathbf{b}}$ with linear part
$L\in\mathrm{GL}(n,2)$ and shift $\mathbf{b}\in\mathbb{Z}_2^n$.

\begin{proposition}\label{proposition: normal form}
The affine group $\mathcal{A}_n$ normalizes $\mathcal{D}^{(\ell)}_n$, so the
unitaries generated by $\{X,\mathrm{CNOT},\mathcal{D}^{(\ell)}_n\}$ have the
semidirect-product structure
$\mathcal{AD}^{(\ell)}_n=\mathcal{D}^{(\ell)}_n\rtimes\mathcal{A}_n$, and
$\mathcal{AD}^{(\ell)}_n\subseteq\mathcal{C}^{(\ell)}_n$ for $\ell\geq2$.
\end{proposition}

\begin{proof}
An element $A\in\mathcal{A}_n$ relabels the computational basis,
$A\ket{\mathbf{x}}=\ket{L\mathbf{x}+\mathbf{b}}$, so conjugation by $A$ sends a
diagonal gate to a diagonal gate. For a generating gate $U_{q,\mathbf{a}}$,
\begin{equation}
    A\,U_{q,\mathbf{a}}\,A^{-1}=\sum_{\mathbf{x}\in\mathbb{Z}_2^n}
    \exp\!\Bigl[\tfrac{2\pi\mathrm{i}}{2^{q}}
    \prod_{j=1}^{n}\bigl(L^{-1}(\mathbf{x}+\mathbf{b})\bigr)_j^{a_j}\Bigr]
    \ket{\mathbf{x}}\!\bra{\mathbf{x}}.
    \label{equation: conjugation by xcnot}
\end{equation}
Since $A$ is a Clifford operator, by induction \cite{zeng2008semi} it is easy to see that $A$ also preserves each Clifford-hierarchy level,
$A\,\mathcal{C}^{(\ell)}_n\,A^{-1}=\mathcal{C}^{(\ell)}_n$. Since $\mathcal{D}^{(\ell)}_n$
is exactly the diagonal part of $\mathcal{C}^{(\ell)}_n$, these two facts give
$A\,\mathcal{D}^{(\ell)}_n\,A^{-1}=\mathcal{D}^{(\ell)}_n$, so $\mathcal{D}^{(\ell)}_n$
is normal. As the identity is the only diagonal element of $\mathcal{A}_n$, we have
$\mathcal{D}^{(\ell)}_n\cap\mathcal{A}_n=\{I\}$. Finally, for $\ell\geq2$ and
any Pauli $P$, writing $U=D_fA$ gives
$UPU^{\dagger}=D_f\bigl(APA^{\dagger}\bigr)D_f^{\dagger}$. Here $APA^{\dagger}$ is a
Pauli, so conjugation by $D_f$ sends it into $\mathcal{C}^{(\ell-1)}_n$ by
\cref{equation: clifford hierarchy}, giving $\mathcal{AD}^{(\ell)}_n\subseteq\mathcal{C}^{(\ell)}_n$.
\end{proof}

Every $U\in\mathcal{AD}^{(\ell)}_n$ therefore factorizes uniquely in the
\emph{affine--diagonal normal form} $U=D_fA$ (or $U=AD_f$) with $D_f\in\mathcal{D}^{(\ell)}_n$
and $A\in\mathcal{A}_n$.

\subsection{Pauli error propagation and commutators}

How a Pauli error propagates through an affine--diagonal circuit is controlled by
the failure of a diagonal gate to commute with an $X$-type Pauli. For
$D\in\mathcal{D}^{(\ell)}_n$ and a Pauli $P$, with
$[a,b]=aba^{-1}b^{-1}$ the group commutator, propagation of $P$
through $D$ is
\begin{equation}
    D\,P\,D^{-1}=[D,P]\,P.
    \label{equation: X propagation}
\end{equation}
The commutator $[D,P]$ is diagonal \cite{cui2017diagonal}, and lies one level lower since $D$ is in the $\ell$th level, i.e.,
\begin{equation}
    [D,P]\in\mathcal{D}^{(\ell-1)}_n.
    \label{equation: commutator level reduction}
\end{equation}

\begin{proposition}\label{proposition: propagation and commutator}
Let $\ell\geq2$. Every propagated error $UPU^{\dagger}$ with
$U\in\mathcal{AD}^{(\ell)}_n$ and $P\in\mathcal{P}_n$ lies in
\begin{equation}
    \mathcal{XD}^{(\ell-1)}_n=\mathcal{D}^{(\ell-1)}_n\rtimes\langle X\rangle_n,
    \label{equation: XD group}
\end{equation}
and the commutator of any two such operators lies in $\mathcal{D}^{(\ell-2)}_n$,
where $\mathcal{D}^{(0)}_n=\{e^{i\phi}I_n\}$ is the group of global phases.
\end{proposition}

\begin{proof}
For the propagation, write $U=D^{(\ell)}A$ in the affine--diagonal normal form of
\cref{proposition: normal form}. Since $A$ is Clifford,
$APA^{\dagger}=P'$ for some
$P'\in\mathcal{P}_n$. By \cref{equation: commutator level reduction},
\begin{equation}
    UPU^{\dagger}
    =D^{(\ell)}P'(D^{(\ell)})^{\dagger}
    =D^{(\ell-1)}X'
\end{equation}
for some $X'\in\langle X\rangle_n$. Therefore the propergated errors is an element of $\mathcal{XD}^{(\ell-1)}_n$.

For the commutator, let $D_1,D_2\in\mathcal{D}_n^{(\ell-1)}$ and
$X_1,X_2\in\langle X\rangle_n$. Since diagonal operators commute and
$X_1X_2=X_2X_1$, we have
\begin{equation}
    [D_1X_1,D_2X_2]
    =[D_1,X_2][X_1,D_2]\in\mathcal{D}_n^{(\ell-2)}.
\end{equation}
Each factor lies in $\mathcal{D}_n^{(\ell-2)}$ by
\cref{equation: commutator level reduction}.
\end{proof}

As a concrete example, take $\ell=3$, i.e.\ the gate set
$\{X,\mathrm{CNOT},Z,\mathrm{CZ},T,\mathrm{CS},\mathrm{CCZ}\}$. The level reduction
\cref{equation: commutator level reduction} is exhibited by the single-gate
commutators
\begin{equation}
    [T,X]=e^{-\mathrm{i}\pi/4}S,
    \qquad
    [\mathrm{CCZ},X_1]=\mathrm{CZ}_{23},
    \nonumber
\end{equation}
which send the level-$3$ gates $T,\mathrm{CCZ}$ to the level-$2$ gates
$S,\mathrm{CZ}$. Propagated Pauli operators therefore lie in $\mathcal{XD}^{(2)}_n$, with
elements of the form
\begin{equation}  \omega^{p}\prod_{j>i}\mathrm{CZ}_{ij}^{c_{ij}}\bigotimes_{k=1}^{n} S_k^{b_k} X_k^{a_k}
    ,
    \label{equation: XD2 element}
\end{equation}
where $\omega=e^{2\pi\mathrm{i}/8}$, $p\in\mathbb{Z}_8$, $\mathbf{a}\in\mathbb{Z}_2^n$,
$\mathbf{b}\in\mathbb{Z}_4^n$, and $\mathbf{c}\in\mathbb{Z}_2^{\binom{n}{2}}$. Their
pairwise commutators lie in $\mathcal{D}^{(1)}_n$, i.e.\ they are $Z$-type Pauli
operators. The failure of two propagated Pauli operators to commute in $\mathcal{A}\mathcal{D}^{(3)}$ circuits are $Z$ Pauli strings, the structural property exploited by the syndrome-sampling algorithm
of the following section.

\section{Syndrome sampling}\label{sec:syndrome_sampling}

In this section, we introduce the algorithm that performs efficient syndrome sampling. In non-Clifford circuits, Pauli operators propagate to non-Pauli operators, and the syndrome measurement results are no longer deterministic. However, for $\mathcal{A}\mathcal{D}^{(3)}$ circuits, the distribution of syndrome for a given circuit-level Pauli noise is relatively simple, i.e., lives randomly in an affine subspace as we show now.

\subsection{Syndrome sampling and noise estimation in $\mathcal{AD}^{(3)}$ circuits\label{subsection: syndrome sampling}}
We show that we can still efficiently evaluate the syndrome if the affine-diagonal circuit is in the 3rd level of the Clifford hierarchy, and evaluate the syndrome expectation value if the circuit is in the 4th level of the Clifford hierarchy in this family.
\begin{syndromerestatement}[Syndrome sampling, restated]
Consider a stabilizer input state evolved by an
$\mathcal{AD}^{(\ell)}$ circuit and followed by commuting stabilizer
measurements. For a fixed configuration of circuit-level Pauli faults,
stabilizer expectation values can be evaluated efficiently for
$\ell\leq4$. For $\ell\leq3$, the complete syndrome can moreover be
sampled exactly in polynomial time using \cref{algorithm: syndrome sampling concise}.
\end{syndromerestatement}

\begin{proof}
    Consider an initial logical state $\rho_{\text{ini}}=\ket{\overline{\psi}}\bra{\overline{\psi}}$ evolving under the logical circuit to the ideal final state $\rho_f=\overline{U}\rho_{\text{ini}}\overline{U}^{\dagger}$ and a stabilizer element $M$ to be measured at the end of the circuit. We in addition insert a collection of sampled circuit Pauli errors $\{e_i\}, i=0.5,...,T+0.5$, where $e_{i-0.5}$ is the Pauli error right before the $i$th circuit layer. Then the probability distribution of measuring $M=\pm 1$ is given by $\frac{1\pm\langle M\rangle}{2}$ and the expectation value $\langle M\rangle$ is given by
\begin{align}
    \langle M\rangle=
    \Tr [M(\prod^{T-0.5}_{i=-0.5}e^f_{T-i})\rho_f(\prod^{T-0.5}_{i=-0.5}e^f_{T-i})^{\dagger}],
\end{align}
where $e^f_{j-0.5}=U_{j,T+1}e_{j-0.5}U^{\dagger}_{j,T+1}$ is the circuit-level fault right before layer $j$ propagated to the end. Equivalently, one could propagate errors and measurements to the beginning of the circuit, resulting in measuring $M_{\text{ini}}=\overline{U}^{\dagger}M\overline{U}$:
\begin{align}
    \langle M\rangle&=\Tr [M_{\text{ini}}(\prod^{T-0.5}_{i=-0.5}e^{\text{ini}}_{T-i})\rho_{\text{ini}}(\prod^{T-0.5}_{i=-0.5}e^{\text{ini}}_{T-i})^{\dagger}]\\
    &=\Tr[[(\prod^{T-0.5}_{i=-0.5}e^{\text{ini}}_{T-i})^{\dagger},M_{\text{ini}}]\,\rho_{\text{ini}}]\\
    &=\Tr[[(e^{\text{ini}})^{\dagger},M_{\text{ini}}]\,\rho_{\text{ini}}]\label{equation: initial_origin_commutator}
\end{align}
where $[a,b]=aba^{-1}b^{-1}$ is the commutator in group theory and $e^{\text{ini}}_{j-0.5}=U_{1,j}^{\dagger}e_{j-0.5}U_{1,j}$ is the circuit-level fault right before layer $j$ propagated to the beginning. We let $e^{\text{ini}}$ denote the product of all Pauli operators propagated to the beginning of the circuit. If the circuit is in $\mathcal{AD}^{(4)}$ then $[(e^{\text{ini}})^{\dagger},M_{\text{ini}}]$ is some Clifford operator according to \cref{proposition: propagation and commutator}. Therefore, $\cref{equation: initial_origin_commutator}$ can be calculated using Clifford circuit simulators. Now we consider circuits in $\mathcal{AD}^{(3)}$, then $[(e^{\text{ini}})^{\dagger},M_{\text{ini}}]$ is just diagonal Pauli operators. As a result $\langle M\rangle$ is either $\pm 1,0$ if the initial state $\rho_{\text{ini}}$ is a stabilizer state.

Let $\mathcal{M}=\langle M_1, M_2, ..., M_m\rangle$ be the abelian Pauli subgroup generated by the measured stabilizer generators. We use the binary vector $\vec{r}\in\mathbb{Z}_2^m$ to label each element $M^{(\vec{r})}=\prod\limits^m_{i}M^{r_i}_i$ in $\mathcal{M}$. From \cref{equation: initial_origin_commutator} we see that the element in $\mathcal{M}$ that stabilizes the output noisy state $\rho_f$ forms a linear subgroup $\mathcal{M}^{\ast}\leq \mathcal{M}$ and their labels $\vec{r}^*$ is in a corresponding linear subspace $V$ of $\mathbb{Z}_2^m$ with dimension $\dim V=k\leq m$.

To sample the syndrome given the circuit-level Pauli noise, we first observe that there is an invertible map between the syndrome expectation of all the stabilizer elements and the probability distribution of the syndrome bit-strings. Given a circuit where $m$ stabilizer generators are measured, we have
\begin{align}
   \langle M^{(\vec{r})}\rangle&=\sum_{\vec{s}\in\{0,1\}^{m}} (-1)^{\vec{r}\cdot\vec{s}}\Pr(\vec{s})\\
      \Pr(\vec{s})&=\frac{1}{2^m}\sum_{\vec{r}\in\{0,1\}^{m}} (-1)^{\vec{r}\cdot\vec{s}}\langle M^{(\vec{r})}\rangle.
\end{align}
where $\vec{s}\in\mathbb{Z}^m_2$ is the syndrome binary vector. Note that $\langle M\rangle=0$ if $\pm M$ is not a stabilizer of the final noisy state, therefore we can focus on the subspace $V$ of the deterministic stabilizers:
\begin{equation}
    \Pr(\vec{s})=\frac{1}{2^m}\sum_{\vec{r}^{\ast}\in V} (-1)^{\vec{r}\cdot\vec{s}}\langle M^{(\vec{r}^{\ast})}\rangle.
\end{equation}
Let $\vec{s}_0\in\mathbb{Z}_2^m$ be any possible syndorme, then it must satisfy:
\begin{equation}
    (-1)^{\vec{r}^{\ast}\vec{s}_0}=\langle M^{(\vec{r}^{\ast})}\rangle, \forall \vec{r}^{\ast}\in V
\end{equation}
since $M^{(\vec{r}^{\ast})}$, upto a sign, is a stabilizer.
Therefore, we have the probability distribution as a Hadamard transformation of a constant function with parameter shifted by $\vec{s}_0$:
\begin{align}
    \Pr(\vec{s})&=\frac{1}{2^m}\sum_{\vec{r}\in V} (-1)^{\vec{r}\cdot\vec{s}}\langle M^{(\vec{r})}\rangle.\\
    &= \frac{1}{2^m}\sum_{\vec{r}\in V} (-1)^{\vec{r}\cdot(\vec{s}+\vec{s}_0)}\\
    &= \frac{1}{2^{m-k}}\mathbf{1}_{V^{\perp}}(\vec{s}+\vec{s}_0).
\end{align}
where $\mathbf{1}_{V^{\perp}}$ is the indicator function of the orthogonal complement $V^{\perp}$ of $V$ and $|V^{\perp}|=2^{m-k}$.  As a result, the syndrome distribution is uniform in an affine subspace with the linear subspace $V^{\perp}$ shifted by $\vec{s}_0$.

To find any valid syndrome $\vec{s}_0$, we first find a basis of $V$ by looking at the syndrome of $[(e^{\text{ini}})^{\dagger}, M^{\text{ini}}_i]$, since $\vec{r}^{\ast}$ is in $V$ iff $[(e^{\text{ini}})^{\dagger}, (M^{\text{ini}})^{(\vec{r}^{\ast})}]$ has trivil syndrome. We define matrices $R$ and $H$ where rows of $R$ and columns of $H$ are the symplectic representation of the Pauli $[(e^{\text{ini}})^{\dagger}, M_i^{\text{ini}}]$  and initial state stabilizer generators, respectively. Denoting $\odot$ as the matrix product using the symplectic inner product in $\mathbb{Z}_2^n\times\mathbb{Z}_2^n$, we have the $i$th row of $R\odot H$ represents the syndrome of $[(e^{\text{ini}})^{\dagger}, M_i^{\text{ini}}]$ for the inital stabilizer state $\rho^{\text{ini}}$. Define the generator matrix $G_V\in\mathbb{Z}_2^{k\times m}$ of the subspace $V$ where each row is a basis vector $\vec{v}_i\in\mathbb{Z}_2^m$ of $V$. The rows of $G_V$ are obtained by computing a basis of $\ker[(R\odot H)^T]$ using Gaussian elimination. Here $R\odot H$ is the matrix $B$ in \cref{algorithm: syndrome sampling concise}.
Define $\vec{a}\in\mathbb{Z}_2^k$ such that $(-1)^{\vec{a}}=(\langle M^{\vec{v}_i}\rangle)_i$, then $\vec{s}_0$ must satisfy $G_V \vec{s}_0=\vec{a}$ so we can choose,
\begin{equation}
    \vec{s}_0=G_V^{g}\vec{a}
\end{equation}
where $G_V^g$ can be any generalized inverse satisfying $G_VG_V^gG_V=G_V$. Finally, choose a matrix $N$ whose columns form a basis of $\ker G_V=V^{\perp}$ and sample $\vec{s}=\vec{s}_0+N\vec{z}$ with $\vec{z}\in\mathbb{Z}_2^{m-k}$ uniform, as in \cref{algorithm: syndrome sampling concise} in the main text.
\end{proof}

The sampling procedure in \cref{algorithm: syndrome sampling concise} in the main text can be done by combining Gaussian elimination with the Tableau simulator in stim \cite{gidney2021stim}.

\section{Efficient benchmark}\label{sec:efficient_benchmark}

\subsection{Logical fidelity}\label{sec:logical_fidelity}
We consider the logical circuit sampling tasks and estimate the logical fidelity at the end of the circuit before the final syndrome measurement and logical measurements. A notion of logical output fidelity $F_{\mathrm{out}}$ was defined in \cite{hangleiter2025fault}, which is the overlap between the ideal output state and the noisy state after being corrected back to the code space by the recovery map $\mathcal{R}$. In the context of CSS code, one typically performs a transversal measurement at the end and therefore syndrome in only one basis is available. In this case, the recovery map is based on the available syndrome, and the remaining syndrome, which is virtually performed by a perfect syndrome measurement:
\begin{equation}
    F_{\mathrm{out}}=\Tr[\mathcal{R}(\rho)\sigma].
\end{equation}
A decoding algorithm is implicitly in the recovery map $\mathcal{R}$. We require the decoding and recovery procedure of the remaining syndrome to be independent of the decoding of the available syndrome, so that the resulting logical fidelity reflects the sample quality. This is because the remaining syndrome, which is not available during the logical sampling, could in principle improve the sampling quality if considering correlated decoding.

\begin{theorem}[Logical fidelity estimation]
\label{theorem: logical fidelity estimation}
Under the assumptions of the main-text logical-fidelity section, for $U\in\mathcal{AD}^{(3)}_n$ and efficiently
sampleable circuit-level Pauli noise, the recovered output fidelity admits
an unbiased Monte Carlo estimator with polynomial cost per attempted
sample in the total number of data and ancilla qubits and the circuit
size, excluding classical decoding and evaluation of the postselection rule.
\end{theorem}

\begin{proof}
We include the data and stabilizer ancillas in a single system, with
joint input and output reference projectors $\Pi_{0}^{\mathrm{in}}$ and
$\Pi_{0}^{\mathrm{out}}$. For product reference subspaces,
\begin{equation}
    \Pi_{0}^{\mathrm{in}}=\Pi_{0}^{D,\mathrm{in}}\otimes\Pi_A^{\mathrm{in}},
    \qquad
    \Pi_{0}^{\mathrm{out}}=\Pi_{0}^{D,\mathrm{out}}\otimes\Pi_A^{\mathrm{out}}.
\end{equation}
The ideal circuit satisfies $U\Pi_{0}^{\mathrm{in}}U^\dagger=\Pi_{0}^{\mathrm{out}}$,
and the joint states obey
$\Pi_{0}^{\mathrm{in}}\ket{\overline{\psi}}=\ket{\overline{\psi}}$ and
$\ket{\overline{\phi}}=U\ket{\overline{\psi}}$.
Let $\mathcal M$ denote the commuting output stabilizer checks used in the
recovery, including any virtual checks. As in the main text, virtual
restoration of the ancillas to their ideal pure state identifies the joint
overlap with the data-block fidelity when the ideal output factorizes.

For a sampled Pauli fault configuration $e$ of probability $p_e$, write the
faulty circuit as $U_e=UC_e$, propagating its errors to the input. For
$U\in\mathcal{AD}^{(3)}_n$, $C_e$ is Clifford by
\cref{proposition: propagation}. Before syndrome measurement,
\begin{equation}
    \rho_{\mathrm{out}}=\sum_e p_e\,
    U_e\ket{\overline{\psi}}\bra{\overline{\psi}}U_e^\dagger.
\end{equation}
The Pauli recovery and its input-frame representation satisfy
\begin{equation}
    P_{\mathbf{s}}\Pi_{\mathbf{s},\mathcal M}^{\mathrm{out}}
    =\Pi_{0,\mathcal M}^{\mathrm{out}}P_{\mathbf{s}},
    \qquad
    C_{\mathbf{s}}=U^\dagger P_{\mathbf{s}}U,
    \qquad P_{\mathbf{s}}U=UC_{\mathbf{s}}.
\end{equation}
We can therefore derive the fidelity step by step:
\begin{align}
    F_{\mathrm{out}}
    &=\sum_{e,\mathbf{s}}p_e
    \left|\bra{\overline{\phi}}P_{\mathbf{s}}
    \Pi_{\mathbf{s},\mathcal M}^{\mathrm{out}}
    U_e\ket{\overline{\psi}}\right|^2\nonumber\\
    &=\sum_{e,\mathbf{s}}p_e
    \left|\bra{\overline{\phi}}P_{\mathbf{s}}
    \Pi_{\mathbf{s},\mathcal M}^{\mathrm{out}}
    UC_e\ket{\overline{\psi}}\right|^2\nonumber\\
    &=\sum_{e,\mathbf{s}}p_e
    \left|\bra{\overline{\phi}}\Pi_{0,\mathcal M}^{\mathrm{out}}
    P_{\mathbf{s}}UC_e\ket{\overline{\psi}}\right|^2\nonumber\\
    &=\sum_{e,\mathbf{s}}p_e
    \left|\bra{\overline{\phi}}\Pi_{0,\mathcal M}^{\mathrm{out}}
    UC_{\mathbf{s}}C_e\ket{\overline{\psi}}\right|^2\nonumber\\
    &=\sum_{e,\mathbf{s}}p_e
    \left|\bra{\overline{\phi}}UC_{\mathbf{s}}C_e
    \ket{\overline{\psi}}\right|^2\nonumber\\
    &=\sum_{e,\mathbf{s}}p_e
    \left|\bra{\overline{\psi}}U^\dagger U C_{\mathbf{s}}C_e
    \ket{\overline{\psi}}\right|^2\nonumber\\
    &=\sum_{e,\mathbf{s}}p_e
    \left|\bra{\overline{\psi}}C_{\mathbf{s}}C_e
    \ket{\overline{\psi}}\right|^2.
    \label{equation: IQP logical fidelity}
\end{align}
Successive steps propagate the faults to the input, move the syndrome
projector through the recovery, back-propagate $P_{\mathbf{s}}$, absorb
$\Pi_{0,\mathcal M}^{\mathrm{out}}$ into the ideal target bra, substitute
$\bra{\overline{\phi}}=\bra{\overline{\psi}}U^\dagger$, and cancel
$U^\dagger U$. The projector is absorbed because $\mathcal M$ stabilizes
the ideal output. The data syndrome need not be complete.

Both $C_e$ and $C_{\mathbf{s}}$ are Clifford, so each squared overlap
$w_e(\mathbf{s})$ can be computed in $O(n^3)$ time from the stabilizer
tableaux of the two states on $n$ physical qubits~\cite{aaronson2004improved}.
Sampling $e\sim p_e$ and a syndrome of probability
$q_e(\mathbf{s})$ gives the unbiased sample value
$w_e(\mathbf{s})/q_e(\mathbf{s})$. The branch probability is available from
\cref{equation: affine syndrome distribution}. Thus each sample has
polynomial cost in $n$ and the circuit size, excluding the prescribed
classical processing.
\end{proof}

\subsection{Observable expectation}
\label{subsection: observable expectation}

\begin{theorem}[Observable expectation estimation]
\label[theorem]{proposition: observable expectation}
Let $U\in\mathcal{AD}^{(3)}_n$ satisfy
\cref{equation: projector intertwining}. For any
$O\in\mathcal{XD}^{(2)}_n$, the expectation value in
\cref{equation: corrected observable} admits an unbiased Monte Carlo
estimator with polynomial cost per sample under circuit-level Pauli noise.
\end{theorem}

\begin{proof}
Use the joint data--ancilla setup of the preceding section, with initial
stabilizer state $\ket{\Omega}$. Here the syndrome is complete, and the
Pauli correction satisfies
$Q_{\mathbf{s}}\Pi_{\mathbf{s}}^{\mathrm{out}}
=\Pi_{0}^{\mathrm{out}}Q_{\mathbf{s}}$.
With $C_e$ denoting the propagated error as above, define
$C_{\mathbf{s}}=U^\dagger Q_{\mathbf{s}}U$ and
$O_{\mathrm{in}}=U^\dagger O U$. The probability of outcome $\mathbf{s}$
is $q_e(\mathbf{s})=\|\Pi_{\mathbf{s}}^{\mathrm{out}}
UC_e\ket{\Omega}\|^2$.
For a nonzero-probability branch, the corrected expectation is
\begin{align}
    \langle O\rangle_{e,\mathbf{s}}
    &=\frac{
    \bra{\Omega}C_e^\dagger U^\dagger
    \Pi_{\mathbf{s}}^{\mathrm{out}}Q_{\mathbf{s}}^\dagger
    O Q_{\mathbf{s}}\Pi_{\mathbf{s}}^{\mathrm{out}}
    UC_e\ket{\Omega}}
    {q_e(\mathbf{s})}\nonumber\\
    &=\frac{
    \bra{\Omega}C_e^\dagger C_{\mathbf{s}}^\dagger
    \Pi_{0}^{\mathrm{in}}O_{\mathrm{in}}\Pi_{0}^{\mathrm{in}}
    C_{\mathbf{s}}C_e\ket{\Omega}}
    {\bra{\Omega}C_e^\dagger C_{\mathbf{s}}^\dagger
    \Pi_{0}^{\mathrm{in}}C_{\mathbf{s}}C_e\ket{\Omega}}.
    \label{equation: conditioned observable supplement}
\end{align}
The second line uses the correction identity and
$U^\dagger\Pi_{0}^{\mathrm{out}}U=\Pi_{0}^{\mathrm{in}}$.

Both $C_e$ and $C_{\mathbf{s}}$ are Clifford. Moreover, affine--diagonal
conjugation preserves $\mathcal{XD}^{(2)}_n$, so $O_{\mathrm{in}}$ is also
Clifford. Applying $\Pi_{0}^{\mathrm{in}}$ to
$C_{\mathbf{s}}C_e\ket{\Omega}$ is Pauli postselection on a stabilizer
state, and the resulting expectation is a stabilizer-state overlap.
Thus each conditional value is efficiently computable, with normalization
implicit when using a normalized tableau state.

Sampling $e\sim p_e$ and $\mathbf{s}\sim q_e$ using
\cref{algorithm: syndrome sampling concise}, then averaging
$\langle O\rangle_{e,\mathbf{s}}$, gives an unbiased estimator of
\cref{equation: corrected observable}.
For an accepted set $\mathcal A$ of syndrome outcomes, the conditional
expectation is
\begin{equation}
    \langle O\rangle_{\mathrm{acc}}
    =\frac{\displaystyle\sum_e p_e
    \sum_{\mathbf{s}\in\mathcal A}q_e(\mathbf{s})
    \langle O\rangle_{e,\mathbf{s}}}
    {\displaystyle\sum_e p_e
    \sum_{\mathbf{s}\in\mathcal A}q_e(\mathbf{s})}.
    \label{equation: postselected observable expectation}
\end{equation}
It is estimated by averaging a fixed number of accepted sample values.
Rejection sampling incurs the inverse acceptance probability as an
additional overhead. The complete syndrome is needed for the correction
identity used above.
\end{proof}

\subsection{Computational cost}
\label{subsection: computational cost}

Let $n$ be the total number of simulated data and ancilla qubits, $m$
the number of commuting measurement outcomes sampled over the complete
circuit, and $T$ the circuit depth. We assume $\mathcal O(n)$ physical gate
locations per layer and that each physical gate is supported on a constant
number of qubits. Since each location is then involved in only a constant
number of nonidentity Pauli faults under the standard circuit-level noise
model, the number of elementary fault events is
\begin{equation}
    N_{\mathrm f}=\mathcal O(nT).
    \label{equation: number elementary faults}
\end{equation}
For an $\mathcal{AD}^{(3)}$ circuit, the propagated errors are elements of
$\mathcal{XD}^{(2)}_n$ and can be represented as
\begin{equation}
    E=
    \omega^p
    \left(\prod_i S_i^{s_i}\right)
    \left(\prod_{i<j}\mathrm{CZ}_{ij}^{\Gamma_{ij}}\right)
    X^{\mathbf x},
    \label{equation: XD2 computational representation}
\end{equation}
where $\omega=e^{i\pi/4}$, $p\in\mathbb Z_8$,
$\mathbf x\in\mathbb Z_2^n$, $\mathbf s\in\mathbb Z_4^n$, and
$\Gamma$ is a binary graph specifying the CZ factors. Here $S_i$ denotes
the single-qubit phase gate on qubit $i$.
This representation requires $\mathcal O(n+|E(\Gamma)|)$ storage and at
most $\mathcal O(n^2)$ storage in the dense case. Updating a dense
representation through one depth-one circuit layer costs
$\mathcal O(n^2)$.

For each measured stabilizer operator $M_i$, let
$M_i^{\text{ini}}=U^\dagger M_iU$, where $U$ is the ideal circuit.
For each elementary fault event $\lambda$, let $E_\lambda$ be the
corresponding Pauli fault propagated to the input, and define
\[
    K_{\lambda i}=[E_\lambda^\dagger,M_i^{\text{ini}}],
\]
where $[A,B]=ABA^{-1}B^{-1}$ is the group commutator.
Although there are $N_{\mathrm f}m$ possible pairs, we assume that each
measurement has a local backward light cone: for
\[
    \Lambda_i=\{\lambda:K_{\lambda i}\neq I\},
\]
assume that each detector depends on faults within a
constant-depth interval, as in repeated syndrome extraction
with detectors formed from consecutive measurements of the
same conserved check. Since each layer contains
$\mathcal O(n)$ fault locations, each detector has nontrivial
commutators with at most $\mathcal O(n)$ elementary faults. Hence the total number of
nontrivial fault--measurement pairs is
\begin{equation}
    L:=\sum_{i=1}^m|\Lambda_i|=\mathcal O(mn).
    \label{equation: nontrivial commutator count}
\end{equation}
This locality assumption is separate from the affine--diagonal group
structure.

As a one-time preprocessing step, we back-propagate all $m$ measured
operators, all $N_{\mathrm f}=\mathcal O(nT)$ elementary faults, and the
single-qubit Pauli generators used to construct output corrections when
decoding. This costs $\mathcal O[(m+nT)n^2T]$. For each of the
$L=\mathcal O(mn)$ nontrivial pairs, we compute and store the phase and
$Z$-type Pauli label of $K_{\lambda i}$, together with its syndrome relative
to the initial stabilizer state. In a dense representation this costs
$\mathcal O(Ln^2)=\mathcal O(mn^3)$. Thus
\begin{equation}
    C_{\mathrm{pre}}
    =\mathcal O\!\left[(m+nT)n^2T+mn^3\right].
    \label{equation: local preprocessing complexity}
\end{equation}

Consider now a sampled trajectory $a$ containing $w_a$ faults.
The precomputed representations of these faults can be composed in
$\mathcal O(w_an^2)$ time to obtain the effective input error $E_a$.
Let $R_a$ denote the number of nontrivial fault--measurement pairs for this
trajectory. The corresponding $Z$-type commutators and their initial
syndromes are combined in $\mathcal O(R_an)$ time. By
\cref{equation: nontrivial commutator count},
$R_a\leq L=\mathcal O(mn)$.

For this fixed trajectory $a$, let $S_1,\ldots,S_n$ generate the initial
pure stabilizer state $\ket{\Omega}$. In the affine--diagonal frame, write
\begin{equation}
\begin{aligned}
    K_{a,i}
    &:=[E_a^\dagger,M_i^{\text{ini}}]
      =i^{p_i}Z^{\mathbf z_i},
\end{aligned}
\label{equation: trajectory row operators}
\end{equation}
where $D_i$ is diagonal Clifford,
$\mathbf x_i,\mathbf z_i\in\mathbb F_2^n$, and $p_i\in\mathbb Z_4$.
$K_{a,i}$ is its group commutator with the adjoint of the propagated error,
in the order shown in \cref{equation: trajectory row operators}.
Unlike $K_{\lambda i}$, $K_{a,i}$ refers to the complete sampled trajectory.
Define the binary matrix
\begin{equation}
    B_{ij}=\langle K_{a,i},S_j\rangle_{\mathrm{sp}},
    \qquad B\in\mathbb F_2^{m\times n},
    \label{equation: trajectory syndrome matrix}
\end{equation}
whose entries are zero for commuting pairs and one for anticommuting pairs.
Here $S_j$ denotes an initial-state stabilizer generator, rather than the
phase gate appearing in \cref{equation: XD2 computational representation}.

Let the rows of $G$ form a basis of $\ker B^T$. Each row labels a product
of measured checks whose outcome is deterministic. Consequently, the
syndrome bits, defined by assigning outcome $(-1)^{s_i}$ to $M_i$, satisfy
the binary parity equations
\begin{equation}
    G\mathbf s=\boldsymbol\alpha,
    \label{equation: deterministic syndrome parity constraints}
\end{equation}
where each entry $\alpha_\ell\in\mathbb F_2$ is calculated by evaluating
\begin{equation}
    (-1)^{\alpha_\ell}
    =\bra{\Omega}
      \left[
          E_a^\dagger,
          \prod_i (M_i^{\text{ini}})^{G_{\ell i}}
      \right]
      \ket{\Omega}.
    \label{equation: deterministic syndrome parity signs}
\end{equation}
The expectation is evaluated on the initial stabilizer state:
$\alpha_\ell=0$ if it equals $+1$, and $\alpha_\ell=1$ if it equals $-1$.
These parity equations define the affine subspace from which the syndrome is
sampled uniformly.

To determine these parity relations and their signs together, Gaussian
elimination on $B$ tracks $(\mathbf x_i,\mathbf z_i,p_i)$ for each current
row. Adding row $j$ to row $i$ corresponds to replacing
$M_i^{\text{ini}}$ by $M_i^{\text{ini}}M_j^{\text{ini}}$.
The associated commutator satisfies
\begin{equation}
\begin{aligned}
    [E_a^\dagger,M_i^{\text{ini}}M_j^{\text{ini}}]
    &=K_{a,i}M_i^{\text{ini}}K_{a,j}
      (M_i^{\text{ini}})^\dagger\\
    &=i^{p_i+p_j+2\mathbf x_i\cdot\mathbf z_j}
      Z^{\mathbf z_i+\mathbf z_j}.
\end{aligned}
\label{equation: elimination commutator update}
\end{equation}
Accordingly, the row labels update as
$\mathbf x_i\leftarrow\mathbf x_i+\mathbf x_j$,
$\mathbf z_i\leftarrow\mathbf z_i+\mathbf z_j$, and
$p_i\leftarrow p_i+p_j+2\mathbf x_i\cdot\mathbf z_j$ modulo four,
using the old values on the right-hand sides. Each row addition therefore costs $\mathcal O(n)$, and elimination
requires at most $\mathcal O(mn)$ such additions.
For each resulting zero row, its deterministic sign is evaluated from the
tracked commutator and the initial stabilizer tableau. The row operations
are recorded and reversed to recover the original syndrome from the
sampled transformed outcomes. No explicit $m\times m$ row-transformation matrix is needed. The total cost of elimination, sign evaluation, and sampling is therefore $\mathcal O(mn^2)$.

Including the $\mathcal O(nT)$ cost of sampling the local noise channels,
one exact syndrome sample costs
\begin{equation}
    C_{\mathrm{synd}}(a)
    =\mathcal O\!\left(nT+w_an^2+mn^2\right).
    \label{equation: local syndrome sampling complexity}
\end{equation}
For local circuit-level noise, $w_a=\mathcal O(nT)$.
Assuming $m=\mathcal O(nT)$ measured outcomes, the syndrome
sampling cost above is therefore $\mathcal O(n^3T)$ per sample
after preprocessing. For fidelity or observable estimation,
constructing the back-propagated Pauli correction and evaluating
the final stabilizer projection or overlap takes
$\mathcal O(n^3)$ time. Thus one complete Monte Carlo
trajectory costs
\begin{equation}
    C_{\mathrm{traj}}(a)=\mathcal O(n^3T).
    \label{equation: local fidelity trajectory complexity}
\end{equation}
These bounds exclude the runtime of the decoder and any postselection rule.
\subsection{Cost of [[8,3,2]] hIQP simulation subroutines}

\begin{figure}[H]
  \centering
  \includegraphics[width=0.8\textwidth]{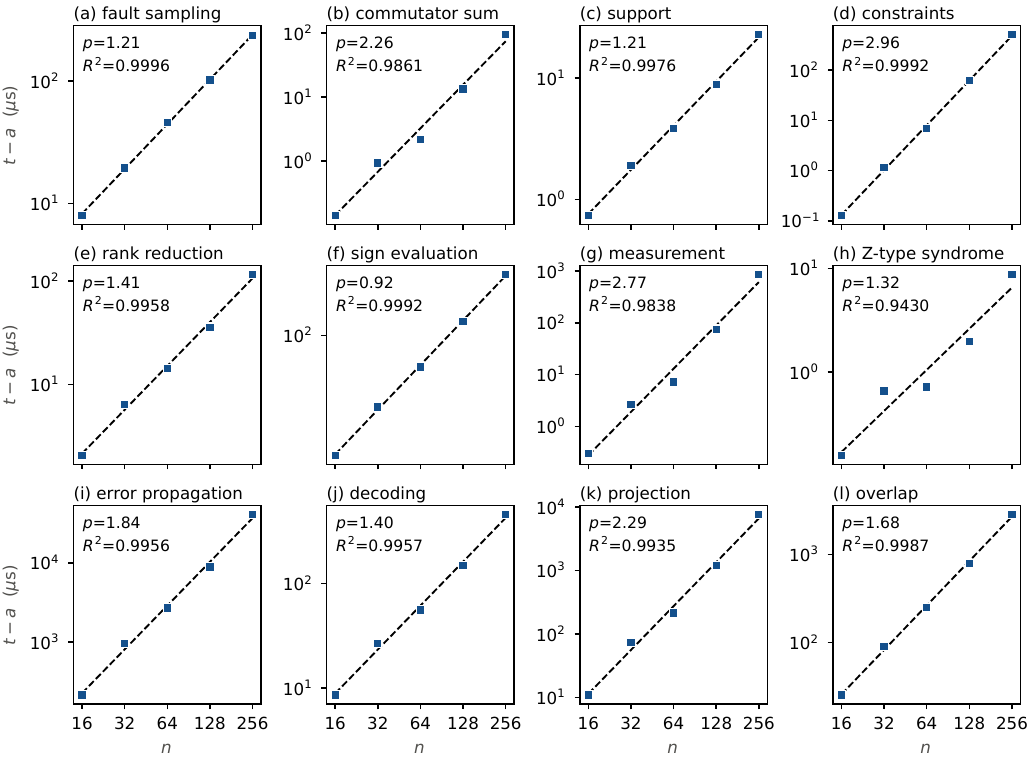}
\caption{Cost of each stage of producing one Monte-Carlo sample, as a function of
    the number of physical qubits $n$, measured on a single CPU for the
    $[[8,3,2]]$ circuit family. Points show the measured time with each stage's
    fixed per-call overhead $a$ subtracted, so that a power law
    $t = a + c\,n^{p}$ appears as a straight line on log--log axes (dashed). The
    fitted exponents range from $p = 0.92$ in (f) to $p = 2.96$ in (d), with
    $R^{2} \ge 0.943$ throughout. Stages (a)--(h) constitute the syndrome
    sampling and (i)--(l) the fidelity evaluation, which accounts for $96\%$ of
    the cost at $n = 256$. The total cost of a sample is the sum of these twelve
    power laws, which is why it is not itself a power law: the stages with the
    smallest exponents dominate at small $n$, the steepest ones at large $n$.
    Each stage is described in \cref{tab:stages}.}

  \label{fig:steps}
\end{figure}

The cost of the simulation separates into a one-time precomputation and a
per-sample Monte-Carlo cost, shown in Fig.~\ref{fig:runtime}. The latter is not a
single power law, its effective slope drifting upward with the system size. To
identify what controls that scaling we time each stage of producing one sample
separately, in Fig.~\ref{fig:steps}: every stage is individually a power law in
$n$, and the total bends only because their exponents differ, so that different
stages dominate at different sizes. Table~\ref{tab:stages} lists the stages.

\begin{table}[H]
\caption{The stages of producing one Monte-Carlo sample and calculating the fidelity, corresponding to the
panels of Fig.~\ref{fig:steps}. $p$ is the fitted exponent of $t = a + c\,n^{p}$
and the last column is the stage's share of the per-sample cost at $n = 256$.
Here $n$ is the number of physical qubits, $m$ the number of measured
stabilizers, and $w$ the number of faults occurring in a given shot.}
\label{tab:stages}
\begin{ruledtabular}
\begin{tabular}{lp{0.52\textwidth}cc}
Panel & Operation & $p$ & Share \\
\colrule
\multicolumn{4}{l}{\textit{Syndrome sampling}} \\
(a) fault sampling   & Draw which faults occur, independently from each elementary noise channel. & 1.21 & $<1\%$ \\
(b) commutator sum   & Accumulate the precomputed commutators of the sampled faults with each measured stabilizer, back-propagated to the circuit input. & 2.26 & $<1\%$ \\
(c) support          & Locate the qubits on which each accumulated commutator acts. & 1.21 & $<1\%$ \\
(d) constraints      & Project the commutators onto the destabilizer basis, giving the linear constraints the syndrome must satisfy. & 2.96 & $1\%$ \\
(e) rank reduction   & Eliminate to find the linearly dependent combinations of measured stabilizers, whose outcomes are deterministic. & 1.41 & $<1\%$ \\
(f) sign evaluation  & Evaluate the signs of those deterministic combinations. & 0.92 & $<1\%$ \\
(g) measurement      & Sample the remaining, non-deterministic outcomes on the stabilizer tableau, one measurement per fault--stabilizer pair with a nontrivial commutator. & 2.77 & $2\%$ \\
(h) Z-type syndrome  & Combine the accumulated commutators with the $Z$-type stabilizers. & 1.32 & $<1\%$ \\
\colrule
\multicolumn{4}{l}{\textit{Fidelity evaluation}} \\
(i) error propagation & Compose the sampled faults into a single effective error at the circuit input and apply it to the stabilizer tableau. & 1.84 & $76\%$ \\
(j) decoding          & Decode the sampled syndrome and apply the resulting Pauli correction. & 1.40 & $1\%$ \\
(k) projection        & Project the corrected state onto the measured stabilizers and reduce it to a canonical stabilizer set. & 2.29 & $14\%$ \\
(l) overlap           & Evaluate the fidelity of the resulting stabilizer state with the ideal one. & 1.68 & $5\%$ \\
\end{tabular}
\end{ruledtabular}
\end{table}

Two of the exponents follow directly from operation counts. Stage~(a) performs one
draw per elementary fault location. Stage~(d) multiplies an $m \times n$ matrix by an
$n \times n/2$ one, where $m$ is the number of measured stabilizers. For this code
family $m = 5n/8$, so the product costs $\Theta(n^{3})$, and the measured
$p = 2.96$ confirms the $m n^{2}$ term of the cost model.

The remaining stages are cheaper than their nominal operation counts, because at the
error rates of interest only $w \approx 0.2$--$6$ faults occur per shot, so the
intermediate objects they act on are almost entirely zero and the corresponding
arithmetic is never executed. In stage~(e), for example, the constraint matrix has one
row per measured $X$-type stabilizer, and a fault anticommutes with only a few of
them: at $n = 256$ half of its $32$ rows vanish identically and its rank is $15$, so
the elimination performs $29$ row operations rather than the $\sim\!10^{3}$ that a
generic matrix of that size would require. The exponents in Table~\ref{tab:stages}
therefore characterise the cost in the regime of physical interest, where the fault
density is low, rather than a worst case.

\section{Benchmarking state/gate teleportation circuits}
\label{sec:teleportation}
\cref{sec:syndrome_sampling} and \cref{sec:efficient_benchmark} study efficient simulations of Clifford-sandwiched affine diagonal circuits, where the circuits can be written as a unitary with measurements at the end of the circuit, including circuits with mid-circuit syndrome extractions. Specifically, given a sampling error configuration, we can efficiently sample syndromes and compute the fidelity conditioned on the syndrome. In this section, we further show that the algorithm we developed can be extended to encoded circuits with state and gate teleportation, where real-time decoding and logical feedforward are involved.

The apparent difficulty with feedforward is that some circuit components depend on certain mid-circuit measurement results, some of which are non-deterministic, e.g., logical measurements. This does not satisfy the circuit-structure condition specified in \cref{sec:syndrome_sampling} and \cref{sec:logical_fidelity} for syndrome sampling and fidelity estimation. However, we show in this section that we could rewrite the circuit into a unitary circuit such that the syndrome sampling and logical fidelity can still be done efficiently without the need to simulate any non-deterministic measurement result.

\subsection{Unitary rewrite}
\begin{figure}[h]
  \centering
  \includegraphics[width=0.5\columnwidth]{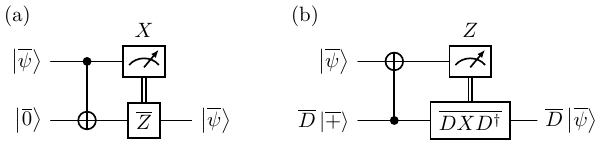}
  \caption{Logical teleportation circuits. Each horizontal line represents an
    entire code block, and all gates and measurements act on the logical level.
    (a)~State teleportation. The input block $\ket{\overline{\psi}}$ controls a
    logical CNOT onto an output block prepared in $\ket{\overline{0}}$. The
    input block is then measured in the logical $X$ basis, and the outcome
    conditions a logical $\overline{Z}$ correction (classical feedforward,
    double line), leaving the output block in $\ket{\overline{\psi}}$.
    (b)~Gate teleportation of a diagonal gate $\overline{D}$. An output block
    prepared in the magic state $\overline{D}\ket{\overline{+}}$ controls a
    logical CNOT onto the input block $\ket{\overline{\psi}}$, which is then
    measured in the logical $Z$ basis. Conditioned on the outcome, the
    correction $\overline{DXD^{\dagger}}$ is applied to the output block,
    leaving it in $\overline{D}\ket{\overline{\psi}}$. For $D$ in the third
    level of the Clifford hierarchy (e.g., $T$ or $\mathrm{CCZ}$), this
    correction is a Clifford operation.}
  \label{fig:teleportation}
\end{figure}

We introduce the basic setting, idea, and notation here. We focus on CSS codes and consider circuits with gate teleportations that implement a logical gate in the $\mathcal{AD}^{(3)}$ group based on the one-bit teleportation circuit. For simplicity, we assume the feedforward operation is ideal or has an identical error rate when the feedforward is not applied and leave the comprehensive analysis incorporating the feedforward error channel to future work. This error-model assumption on feedforward may be justified when it is implemented virtually, or absorbed by subsequent operations \cite{chamberland2018fault, litinski2019game,fowler2012time}.

Consider a single gadget of state teleportation or gate teleportation illustrated in \cref{fig:teleportation}. Let $\ket{\Omega}$ be the initial joint stabilizer state, and let all the unitary gates,
including the magic-resource preparation and the logical CNOT gate of the teleportation gadget, be $U_0$. If $U_0$ is in $\mathcal{AD}^{(3)}$, Then the noisy $U_0$ with Pauli error configuration $a$ can be written as $U_{0,a}=U_0C_a$ with $C_a$ Clifford. We denote the input register and output register as $\mathbb{A}$ and $\mathbb{B}$ respectively.

Let $Q\in\{X,Z\}$ denote the transversal readout basis on $\mathbb A$.
For the physical outcome string $x$, write the extracted syndrome
$s=H_Q x$ and raw logical bit string $l=L_Q x$ for the code syndrome and raw logical
parities, respectively. The rows of $H_Q$ and $L_Q$ specify the
measured code checks and physical logical Pauli representatives.
Let $R_l$ be the physical Pauli or Clifford correction associated with $l$.
Define
\begin{equation}
    F_R=\sum_{l\in\mathbb F_2^k}
       \Pi_l^{Q,L,\mathbb A}\otimes R_l,
    \qquad
    \Pi_l^{Q,L,\mathbb A}
    =\prod_{j=1}^k
      \frac{I+(-1)^{l_j}\overline{Q}^{\mathbb A}_j}{2}.
    \label{eq:gate_FR}
\end{equation}
This corresponds to the operation after backpropagating the classical-controlled feedforward operator before the logical measurement without decoding. Let $\{\ket{x}\}$ be the physical basis states in the basis $Q$. We have,
\begin{equation}
    (\bra{x}_{\mathbb A}\otimes I_{\mathbb B})F_R
    =
    \bra{x}_{\mathbb A}\otimes R_{l},
    \label{eq:gate_FR_property}
\end{equation}
which reproduces the raw feedforward for every physical outcome. With the decoded correction $d_L(s)$ on the logical outcome, the error-corrected logical outcome becomes
$\widehat l(x)=l(x)+d_L(s)$, and the implemented feedforward satisfies
$R_{\widehat l} = R_{d_L(s)}R_l$.  Note that $R_{d_L(s)}$ depends only on $s$, or more generally, all previous syndromes.

Setting $W=F_RU_0$, the unnormalized noisy output becomes
\begin{equation}
\begin{aligned}
    \ket{v_{a,x}}_{\mathbb B}
    &=(\bra{x}_{\mathbb A}\otimes R_{\widehat l(x)})
      U_{0,a}\ket{\Omega}\\
    &=(\bra{x}_{\mathbb A}\otimes R_{d_L(s)})
      W C_a\ket{\Omega}.
\end{aligned}
\label{eq:gate_noisy_rewrite}
\end{equation}
Crucially, the noisy state now only depends on the syndrome information, which can be simulated efficiently (see \cref{sec:syndrome_sampling_feedforward}), independent of the logical outcome.

\begin{lemma}\label{lemma:F_R_AD^{(3)}}
    For CSS code gate teleportation with $Q=Z$ type control and feedforward $R_{l} = \overline{D}\overline{X}^l\overline{D}^{\dagger}\in\mathcal{XD}^{(2)}$, then $F_R\in\mathcal{AD}^{(3)}$.
\end{lemma}
\begin{proof}
    Choose a CNOT decoding network that simultaneously maps the chosen
physical logical operators $\overline Z_{\mathbb A,j}$ to
single-qubit operators $Z_{q_j}$. Apply the corresponding controlled
corrections 
\begin{equation}
    C_{q_i}R_{i} = \ketbra{0}{0}_{q_i}\otimes I + \ketbra{1}{1}_{q_i}\otimes R_i
\end{equation}
where $R_i = \overline{D}\overline{X}_i\overline{D}^{\dagger}$, and undo the decoding network..
Since, $R_i\in \mathcal{XD}^{(2)}$ which can be generated by $X$, $S$, and $CZ$, $C_{q_i}R_{i}$ is generated by $\text{CNOT}$, $CS$, and $CCZ$ and therefore is in the $\mathcal{AD}^{(3)}$, and stays in $\mathcal{AD}^{(3)}$ when conjugated by CNOT gates.
\end{proof}

For the purpose of fidelity calculation, we look at the ideal state after the unitary rewrite, which is of the form
\begin{equation}
    W\ket{\Omega}
    =\ket{\eta}_{\mathbb A}\otimes
      \ket{\overline{\phi}}_{\mathbb B},
    \label{eq:gate_reference}
\end{equation}
where $\ket{\eta}$ is a known pure stabilizer state ($\ket{\overline{0}}$ for state teleportation in \cref{fig:teleportation}(a), and $\ket{\overline{+}}$ for gate teleportation in \cref{fig:teleportation}(b)) and
$\ket{\overline{\phi}}$ is the desired output.

\subsection{Syndrome sampling and logical fidelity with feedforward}\label{sec:syndrome_sampling_feedforward}
The syndrome sampling algorithm requires backpropagation of the errors and measured stabilizer elements in the circuit. However, the circuit itself might depend on the feedforward operation. The syndrome algorithm introduced in \cref{sec:syndrome_sampling} can be modified to sample sequentially according to the circuit layers and sample the subsequent syndrome conditioned on the past syndrome history, taking the decoder correction $R_{d_L(s)}$ into account.

After the unitary rewrite, the ideal output of a teleportation gadget is of the form in \cref{eq:gate_reference}, where it is a tensor product with a known input code stabilizer state that depends on the type of teleportation gadget. Therefore, we could perform a virtual recovery on the input register of each teleportation gadget. The overlap then becomes the fidelity of the output register.

More concretely, after each teleportation gadget for the unitary rewrite, complete the measured $Q$ basis code stabilizer on $\mathbb A$ to the
full stabilizer group of $\ket{\eta}$ in \cref{eq:gate_reference}, and denote the additional
virtual outcomes by $r$. For state teleportation,
$\ket{\eta}=\ket{\overline{0}^{\,k}}$ and
$r=(s_Z,l_Z)$. For the diagonal-gate teleportation considered above,
$\ket{\eta}=\ket{\overline{+}^{\,k}}$ and
$r=(s_X,l_X)$.
These additional logical outcomes are virtual and are distinct from
the original raw logical outcome $l=L_Qx$.

Let $\Pi^{\mathbb A}_{s,r}$ be the resulting rank-one joint
measurement projectors. They satisfy
\begin{equation}
    \sum_r\Pi^{\mathbb A}_{s,r}
    =\Pi_s^{Q,\mathbb A}
    =\sum_{x:H_Qx=s}\ketbra{x}_{\mathbb A}.
    \label{eq:virtual_readout_completeness}
\end{equation}
In the ideal circuit, all checks defining these projectors have
eigenvalue $+1$. For each sampled $(s,r)$, choose a source-only
Pauli correction $P^{\mathbb A}_{s,r}$ restoring the measured state
to $\ket{\eta}$, so that
\begin{equation}
    P^{\mathbb A}_{s,r}\Pi^{\mathbb A}_{s,r}
    =\ketbra{\eta}_{\mathbb A}P^{\mathbb A}_{s,r}.
\end{equation}
The virtual measurement replaces the unused part of the physical
readout. Its additional outcomes $r$ are averaged, not postselected
or used to improve the data recovery.

We in addition perform a perfect final syndrome measurement on the final output code blocks $\mathbb B$
after the residual feedforward $R_{d_L(s)}$.
Let $s_{\mathbb B}$ denote this syndrome,
$\Pi^{\mathbb B}_{s_{\mathbb B}}$ its projector, and
$P^{\mathbb B}_{s_{\mathbb B}}$ the prescribed Pauli recovery, with
\[
    P^{\mathbb B}_{s_{\mathbb B}}
    \Pi^{\mathbb B}_{s_{\mathbb B}}
    =\Pi^{\mathbb B}_0P^{\mathbb B}_{s_{\mathbb B}},
\]
where $\Pi^{\mathbb B}_0$ is the destination code-space projector.
Define the combined Kraus operator
\begin{equation}
    K_{s,r,s_{\mathbb B}}
    =
    P^{\mathbb A}_{s,r}\Pi^{\mathbb A}_{s,r}
    \otimes
    P^{\mathbb B}_{s_{\mathbb B}}
    \Pi^{\mathbb B}_{s_{\mathbb B}}R_{d_L(s)}.
    \label{eq:virtual_recovery_kraus}
\end{equation}
For a fixed fault configuration $a$, write
\[
    \rho'_a
    :=WC_a\ketbra{\Omega}C_a^\dagger W^\dagger,
    \qquad
    \tau^{\mathbb B}_{a,s}
    :=\sum_{x:H_Qx=s}\ketbra{v_{a,x}}_{\mathbb B},
\]
where $\ket{v_{a,x}}$ is the unnormalized physical output defined
in \cref{eq:gate_noisy_rewrite}, already including the decoded
feedforward.
The completeness relation \cref{eq:virtual_readout_completeness}
and basis independence of the partial trace imply
\begin{equation}
\begin{aligned}
    &\sum_r K_{s,r,s_{\mathbb B}}\rho'_a
        K_{s,r,s_{\mathbb B}}^\dagger\\
    &\quad=
    \ketbra{\eta}_{\mathbb A}\otimes
    P^{\mathbb B}_{s_{\mathbb B}}
    \Pi^{\mathbb B}_{s_{\mathbb B}}
    \tau^{\mathbb B}_{a,s}
    \Pi^{\mathbb B}_{s_{\mathbb B}}
    (P^{\mathbb B}_{s_{\mathbb B}})^\dagger.
\end{aligned}
\label{eq:virtual_recovery_equivalence}
\end{equation}
Thus, after averaging $r$, the restored joint state contains exactly
the actual recovered state on $\mathbb B$, conditioned on the same
physical syndromes. This identity holds even when $\rho'_a$ is
entangled.

Let the probability of the retained syndromes at fixed $a$ be
\begin{equation}
    p_a(s,s_{\mathbb B})
    :=
    \sum_r
    \left\|
        K_{s,r,s_{\mathbb B}}WC_a\ket{\Omega}
    \right\|^2.
\end{equation}
The conditional logical fidelity is
\begin{equation}
\begin{aligned}
    F_{a\mid s,s_{\mathbb B}}
    &=
    \frac{1}{p_a(s,s_{\mathbb B})}
    \sum_r
    \left|
        \bra{\Phi}
        K_{s,r,s_{\mathbb B}}WC_a
        \ket{\Omega}
    \right|^2\\
    &=
    \frac{1}{p_a(s,s_{\mathbb B})}
    \sum_r
    \left|
        \bra{\Phi}
        \left(
            P^{\mathbb A}_{s,r}\otimes
            P^{\mathbb B}_{s_{\mathbb B}}R_{d_L(s)}
        \right)
        WC_a\ket{\Omega}
    \right|^2\\
    &=
    \frac{1}{p_a(s,s_{\mathbb B})}
    \sum_r
    \left|
        \bra{\Omega}
        C_{s,r,s_{\mathbb B}}C_a
        \ket{\Omega}
    \right|^2,
\end{aligned}
\label{eq:teleportation_conditional_fidelity}
\end{equation}
where $\ket{\Phi}=W\ket{\Omega}$ is the target state defined in \cref{eq:gate_reference} and
\begin{equation}
\begin{aligned}
    C_{s,r,s_{\mathbb B}}
    &:=
    W^\dagger
    \left(
        P^{\mathbb A}_{s,r}\otimes
        P^{\mathbb B}_{s_{\mathbb B}}R_{d_L(s)}
    \right)W\\
    &=
    U_0^\dagger F_R^\dagger
    \left(
        P^{\mathbb A}_{s,r}\otimes
        P^{\mathbb B}_{s_{\mathbb B}}R_{d_L(s)}
    \right)F_RU_0.
\end{aligned}
\label{eq:teleportation_correction_backpropagation}
\end{equation}
The second equality in
\cref{eq:teleportation_conditional_fidelity} moves the syndrome
projectors through their respective Pauli recoveries and absorbs
the resulting reference projectors into $\bra{\Phi}$.
The final equality uses $\bra{\Phi}=\bra{\Omega}W^\dagger$. For state teleportation, the combined correction $P^{\mathbb A}_{s,r}\otimes
        P^{\mathbb B}_{s_{\mathbb B}}R_{d_L(s)}$ is Pauli.
For the gate-teleportation construction above, it belongs to
$\mathcal{XD}^{(2)}$, since $R_{d_L(s)}\in\mathcal{XD}^{(2)}$.
$W\in\mathcal{AD}^{(3)}$ according to \cref{lemma:F_R_AD^{(3)}}, and therefore
conjugation by $W$ preserves $\mathcal{XD}^{(2)}$.
Hence $C_{s,r,s_{\mathbb B}}$ and $C_a$ are Clifford, and every
squared overlap in
\cref{eq:teleportation_conditional_fidelity}
is efficiently evaluable using stabilizer-state methods. The summation here is over exponentially many $r$s, but this can be estimated using Monte Carlo sampling.

\section{Magic state preparation using Tricycle codes}

\subsection{The $[[27,3,3]]$ code}
\label{sec:tricycle_code}

Tricycle codes are CSS codes constructed from three elements
$a,b,c$ of the group algebra $R=\mathbb{F}_2[G]$ of a finite
Abelian group $G$~\cite{menon2026magic}.
An element in the group algebra can be represented by $h=\sum_{g\in G}h_g g$, with $h_g\in\mathbb{F}_2$,
is a binary sum of group elements. Addition is coefficientwise,
and multiplication extends the group operation bilinearly.
For the code studied here, $G=\mathbb{Z}_3\times\mathbb{Z}_3$
has commuting generators $x,y$ and order $n_G=9$, giving
\begin{equation}
\label{eq:group_algebra}
R\simeq\mathbb{F}_2[x,y]/\langle x^3-1,y^3-1\rangle.
\end{equation}
Thus group-algebra elements are bivariate polynomials with binary
coefficients and exponents reduced modulo three.
We order the group basis by $x^u y^v\mapsto 3u+v$,
with $u,v\in\{0,1,2\}$.

The binary matrix associated with $h\in R$ is its left regular
representation $\rho(h)$, which represents multiplication by $h$
in the group basis:
\begin{equation}
\label{eq:regular_representation}
[\rho(h)]_{g,g'}=h_{g(g')^{-1}},
\qquad g,g'\in G.
\end{equation}
Each monomial is represented by a permutation matrix,
and polynomial sums by binary sums of these matrices.
Writing $\mathbf{A}=\rho(a)$, $\mathbf{B}=\rho(b)$, and
$\mathbf{C}=\rho(c)$, the parity-check matrices for the tricycle code is defined as
\begin{equation}
\label{eq:HXHZ}
\begin{aligned}
H_X&=\begin{bmatrix}
\mathbf{A}^{\mathsf T}&\mathbf{B}^{\mathsf T}&\mathbf{C}^{\mathsf T}
\end{bmatrix},\\
H_Z&=\begin{bmatrix}
\mathbf{B}&\mathbf{A}&0\\
\mathbf{C}&0&\mathbf{A}\\
0&\mathbf{C}&\mathbf{B}
\end{bmatrix}.
\end{aligned}
\end{equation}
Their rows specify the supports of $X$- and $Z$-type stabilizers.
Because $R$ is commutative, $\mathbf{A}$, $\mathbf{B}$, and
$\mathbf{C}$ commute, implying $H_XH_Z^{\mathsf T}=0$.
This construction is the three-fold balanced product over $R$ of the
classical two-term complexes $R\xrightarrow{a}R$,
$R\xrightarrow{b}R$, and $R\xrightarrow{c}R$~\cite{menon2026magic}.
The set of qubits comprises three sectors of $n_G$ qubits each.

We choose the binomials
\begin{equation}
\label{eq:polys}
a=1+x,\qquad b=1+y,\qquad c=x+y,
\end{equation}
yielding the $[[27,3,3]]$ bivariate tricycle
code~\cite{li2025transversal}.
Sector $i\in\{1,2,3\}$ corresponds to $a,b,c$, respectively,
and its qubit labeled by $x^u y^v$ has local index
$9(i-1)+3u+v$.
The matrices $H_X\in\mathbb{F}_2^{9\times27}$ and
$H_Z\in\mathbb{F}_2^{27\times27}$ have ranks $8$ and $16$,
respectively, so $k=27-8-16=3$.
Their row weights are $6$ and $4$, and their column weights are
$2$ and $4$.
Exhaustive enumeration of the logical-operator spaces
$\ker H_Z/\operatorname{rowsp}(H_X)$ and
$\ker H_X/\operatorname{rowsp}(H_Z)$ gives
X distance $d_X=6$ and Z distance $d_Z=3$.
We extract both check types using one ancilla per check and a
six-layer schedule containing $162$ CNOTs.

The $27$ redundant $Z$ checks obey $11$ independent syndrome relations.
For the block-row order in Eq.~\eqref{eq:HXHZ}, the product construction
provides the structured metacheck matrix
\begin{equation}
\label{eq:meta}
H_{\mathrm{meta}}=
\begin{bmatrix}\mathbf{C}&\mathbf{B}&\mathbf{A}\end{bmatrix},
\qquad H_{\mathrm{meta}}H_Z=0,
\end{equation}
which has rank $8$.
In our simulations, we instead use a matrix
$M\in\mathbb{F}_2^{11\times27}$ whose rows form a basis of the
complete left null space of $H_Z$, so that
\begin{equation}
\label{eq:complete_meta}
MH_Z=0,\qquad \ker M=\operatorname{im}H_Z.
\end{equation}
Thus $M\sigma=0$ is necessary and sufficient for a syndrome
$\sigma$ to be compatible with some data-qubit $X$ error.
These complete constraints enforce syndrome consistency during
repair.

\subsection{Depth-two $CCZ$ circuit}
\label{sec:tricycle_ccz}

We construct a code-space-preserving $CCZ$ circuit across three
copies of the code using the cup-product formula of
Ref.~\cite{menon2026magic}.
Let $(\alpha_1,\alpha_2,\alpha_3)=(a,b,c)$ and choose the
preorientations
\begin{equation}
\label{eq:preorientations}
\begin{aligned}
(a^{\mathrm{in}},a^{\mathrm{out}})&=(1,x),\\
(b^{\mathrm{in}},b^{\mathrm{out}})&=(1,y),\\
(c^{\mathrm{in}},c^{\mathrm{out}})&=(x,y),
\end{aligned}
\end{equation}
with empty free parts.
These singleton partitions satisfy the sufficient integrated-Leibniz
conditions for code-space preservation~\cite{menon2026magic}.
For this binomial case, the ordinary and symmetric triple
cup-product constructions agree.

Write $p^{(i)}$ for the physical-qubit basis vector labeled by
$p\in G$ in sector $i$.
The trilinear function $f_{\mathrm{CCZ}}$ is specified on these
basis vectors by
\begin{equation}
\label{eq:fccz}
\begin{aligned}
&f_{\mathrm{CCZ}}(p^{(i)},q^{(j)},r^{(k)})\\
&\quad=\Delta_{ijk}\Bigl|
 r\alpha_i^{\mathrm{in}}\alpha_j^{\mathrm{in}}
 \cap q\alpha_i^{\mathrm{in}}\alpha_k^{\mathrm{out}}
 \cap p\alpha_j^{\mathrm{out}}\alpha_k^{\mathrm{out}}
\Bigr|\bmod 2,
\end{aligned}
\end{equation}
where $\Delta_{ijk}=1$ when $i,j,k$ are pairwise distinct and is
zero otherwise.
Products are evaluated in $R$, and intersections and cardinalities
refer to their supports in the group basis.
A physical $CCZ$ acts on the corresponding qubits of the first,
second, and third code blocks precisely when
$f_{\mathrm{CCZ}}(p^{(i)},q^{(j)},r^{(k)})=1$.
Extending $f_{\mathrm{CCZ}}$ trilinearly to binary strings gives
\begin{equation}
\label{eq:ccz_phase}
U_{\mathrm{CCZ}}\lvert\mathbf{u},\mathbf{v},\mathbf{w}\rangle
=(-1)^{f_{\mathrm{CCZ}}(\mathbf{u},\mathbf{v},\mathbf{w})}
\lvert\mathbf{u},\mathbf{v},\mathbf{w}\rangle.
\end{equation}

For Eq.~\eqref{eq:preorientations}, the three supports in
Eq.~\eqref{eq:fccz} are singletons.
A gate is therefore present exactly when $i,j,k$ are pairwise
distinct and
\begin{equation}
\label{eq:ccz_singleton_condition}
r\alpha_i^{\mathrm{in}}\alpha_j^{\mathrm{in}}
=q\alpha_i^{\mathrm{in}}\alpha_k^{\mathrm{out}}
=p\alpha_j^{\mathrm{out}}\alpha_k^{\mathrm{out}}.
\end{equation}
There are nine solutions for each of the six sector permutations,
giving $54$ physical $CCZ$ gates on $81$ data qubits.
The sector permutations $(1,2,3)$, $(2,3,1)$, and $(3,1,2)$ form
the first layer, and $(1,3,2)$, $(2,1,3)$, and $(3,2,1)$ form
the second.
Each layer contains $27$ mutually disjoint gates, and every qubit
participates once per layer, establishing the optimal depth of two
for this gate set.

We also verify code-space preservation directly:
\begin{equation}
\label{eq:ccz_coboundary}
\begin{aligned}
f_{\mathrm{CCZ}}(s,u,v)&=0,\\
f_{\mathrm{CCZ}}(u,s,v)&=0,\\
f_{\mathrm{CCZ}}(u,v,s)&=0,
\end{aligned}
\end{equation}
for all $s\in\operatorname{rowsp}(H_X)$ and $u,v\in\ker H_Z$.
Consequently, $f_{\mathrm{CCZ}}$ induces a well-defined trilinear
form on the logical $X$-operator classes, and its evaluation on
logical $X$ representatives specifies the logical $CCZ$
connectivity~\cite{menon2026magic}.
Applied to $\lvert\overline{+,+,+}\rangle^{\otimes3}$, with all
nine logical qubits initialized in $\lvert+\rangle$, the circuit
prepares an encoded hypergraph magic state.

\subsection{Details of the simulation of the magic-state preparation}
\label{sec:simulation-details}
Each code block is prepared by initializing all $n_{\mathrm{code}}=27$ data qubits in $\ket{+}$ and
running one round of the depth-$6$ syndrome-extraction circuit under circuit-level
depolarizing noise, i.e., $p_2=10^{-3}$ per CNOT. On $\ket{+}^{\otimes n_{\mathrm{code}}}$ the $X$-checks are
deterministic, so we postselect on a trivial $X$ syndrome. The $Z$-checks are random and
are instead fixed in a single shot---the measured $Z$ syndrome is first repaired using
the complete metacheck set of Sec.~\ref{sec:tricycle_code}, then decoded by BP+OSD and the
inferred $X$ correction applied. What survives is a residual data Pauli $E$, the exact
deviation of the prepared block from $\ket{\overline{+}}$. We obtain $E$ from stabilizer simulation. Every sample of the protocol draws $n_{\mathrm{prep}}$ fresh, independent
residuals, one per prepared block. The residuals are then injected into the $CCZ$ step as faults at $t=0$. Each physical $CCZ$ then fails with total
probability $p_{CCZ}=2\times 10^{-3}$, drawn from the $7$ pure-$Z$ Pauli errors with probability $\frac{10}{11}p_{CCZ}$  and $56$ remaining three-qubit Paulis with probability $\frac{1}{11}p_{CCZ}$.

After the transversal $CCZ$ on the $n_{\mathrm{blk}}=3$ data blocks, one further
$\ket{\overline{+}}$ block is prepared by the same noisy routine and coupled to the data
by transversal CNOTs with the data as control. In the \emph{shared} variant (protocol E) a single
ancilla block is coupled to each data block in turn, giving $n_{\mathrm{blk}}$
sequential layers of $n_{\mathrm{code}}$ CNOTs. In the Steane style syndrome extraction (protocol D), one ancilla per data block is
used and all $n_{\mathrm{blk}}n_{\mathrm{code}}$ CNOTs occur in one layer.

We do not simulate the destructive transversal $Z$ readout of the ancilla, whose
single-qubit outcomes are random on $\ket{\overline{+}}$. The same deterministic
information is carried by the ancilla's $Z$ \emph{stabilizers}: back-propagating
$S_z^{\mathrm{anc}}$ through the transversal CNOT (ancilla as target) gives
$S_z^{\mathrm{anc}}S_z^{\mathrm{data}}$, whose outcome, in the noiseless extraction case, is
\begin{equation}\label{eq:steane-syn}
  H_Z\big(x^{\mathrm{anc}}\oplus x^{\mathrm{data}}\big),
\end{equation}
i.e. exactly the Steane syndrome. We therefore measure these $r_Z$ operators virtually
and postselect them on $+1$. With a shared ancilla, Eq.~\eqref{eq:steane-syn} carries the
\emph{sum} over data blocks, so identical $X$ errors in two blocks cancel and are
invisible to the Steane check---though not to the data $X$ sector---while an ancilla $Z$
fault propagates back onto every data block whose CNOT follows it.

\paragraph{Restoring the ancilla block.}
The ancilla must be returned exactly to $\ket{\overline{+}}$, or it remains correlated
with the data and contaminates the reported fidelity. In the final frame it carries
$X^{\alpha}Z^{\beta}$. Postselecting Eq.~\eqref{eq:steane-syn} forces $H_Z\alpha=0$, so
$X^{\alpha}\in\langle\,X\text{-stabilizers},\,X\text{-logicals}\rangle$ and already acts
trivially on $\ket{\overline{+}}$. What remains is $Z^{\beta}$, which acts trivially only
if $Z^{\beta}$ lies in the $Z$-stabilizer group. Its class in that quotient is precisely
its commutation syndrome $\sigma$ with \emph{all} $X$-type generators. We therefore also
measure virtually the $r_X$ $X$-stabilizers \emph{and} the $K$ $X$-logicals of the
ancilla---$r_X+K=11$ operators here---and apply a $Z$-type Pauli $Z^{G\sigma}$, where $G$
is any right inverse of the matrix of $X$-type generators. Any solution suffices: two
differ by a $Z$ operator commuting with every $X$-type generator, i.e. by a
$Z$-stabilizer, which fixes $\ket{\overline{+}}$.

Once restored, the ancilla is an exact tensor factor of
both the corrected and the ideal state, so the reported fidelity is the fidelity of the
data blocks alone. For the shared gadget this leaves $35$ measured $X$-type and $64$
measured $Z$-type operators, of which $40$---the $24$ data $X$-stabilizers and the $16$
Steane $Z$-checks---are postselected. The remaining data $Z$ syndrome is corrected by the
same per-block BP+OSD decoder.

Finally, we explain how the acceptance fraction is calculated. Acceptance fraction factorizes into the preparation and the post-$CCZ$ stage. Let $\alpha_{\mathrm{prep}}$
be the per-block acceptance of the $X$-postselected $\ket{\overline{+}}$ preparation. A
sample requires all $n_{\mathrm{prep}}$ blocks to pass independently---$n_{\mathrm{prep}}=3$
with no round, $4$ for the shared gadget ($3$ data $+\,1$ ancilla) and $6$ for the
per-block gadget---so the preparation contributes $\alpha_{\mathrm{prep}}^{\,n_{\mathrm{prep}}}$.
The second stage rate $\alpha_{\mathrm{post}}$ is the fraction of prepared samples passing every
postselection after the $CCZ$, and is measured conditioned on the preparation
having been accepted. The total acceptance is therefore
\begin{equation}\label{eq:accept}
  \alpha_{\mathrm{acc}} \;=\; \alpha_{\mathrm{prep}}^{\,n_{\mathrm{prep}}}\;\alpha_{\mathrm{post}}.
\end{equation}
For the shared gadget at $p_2=10^{-3}$, for instance, $\alpha_{\mathrm{prep}}=0.917$ and
$\alpha_{\mathrm{post}}=0.699$ give $\alpha_{\mathrm{acc}}=0.917^{4}\times0.699=0.495$.

\paragraph{Approximate simulation.}
We compare the exact simulation with the approximate simulation. We adapt the approximate simulation used in \textcite{menon2026magic}, but we simulate all three data blocks rather than a single block in their work. The only difference between the approximate simulation in this work and the exact simulator, then, is that we skip the CCZ circuit and twirl the CZ errors that arise from $X$ errors propagating through it. That is, whenever a CZ error occurs, we replace it with an error drawn uniformly at random from $\{II, IZ, ZI, ZZ\}$. Ignoring the CCZ gates in the circuit makes this simulation effectively an $ X$-basis memory experiment that detects logical Z errors. This is believed to be a good approximation for obtaining the output logical fidelity because, on the one hand, logical Z errors commute with the logical $CCZ$ circuit and can act directly on the initial $\ket{\overline{+}}^{\otimes k}$ state. On the other hand, tricycle codes have biased distance \cite{menon2026magic}, with X distance $ d_X$ much larger than Z distance, making the logical X error rate negligible compared with the logical Z error rate.

\subsection{Fault tolerance of depth-two CCZ circuit}
\label{sec:ccz-single-faults}

\begin{figure}[H]
    \centering
    \includegraphics[width=0.8\textwidth]{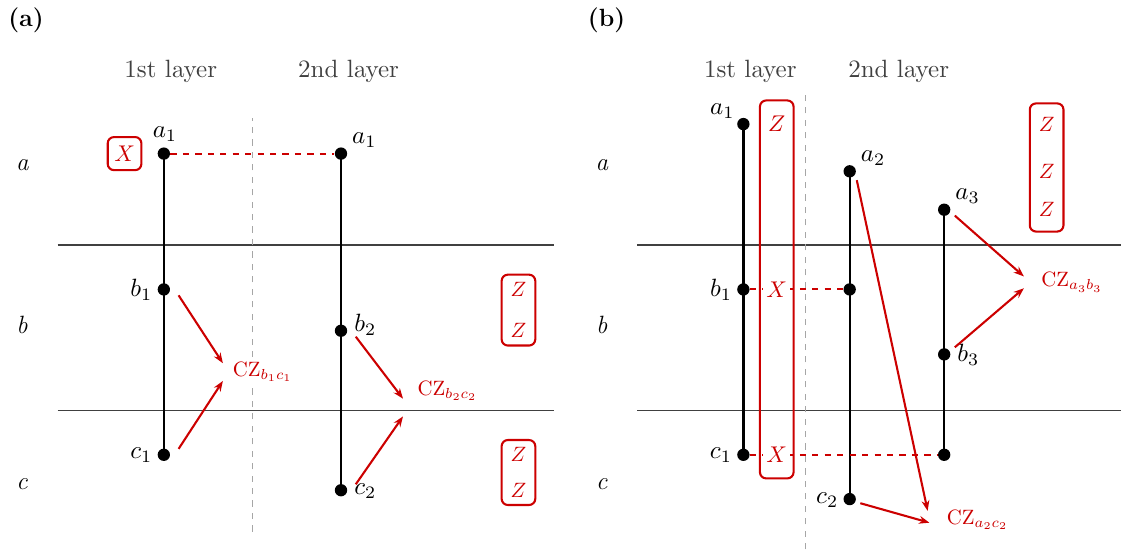}
    \caption{
    Error propagation through the depth-two CCZ circuit on code
    blocks $a$, $b$, and $c$.
    (a) An input error $X_{a_1}$ propagates through two CCZ gates,
    acquiring $\mathrm{CZ}_{b_1c_1}\mathrm{CZ}_{b_2c_2}$.
    Their Pauli expansion contains a term with weight-two $Z$
    errors in both blocks $b$ and $c$, which can be uncorrectable.
    (b) A weight-three fault $Z_{a_1}X_{b_1}X_{c_1}$ after a
    first-layer CCZ gate acquires
    $\mathrm{CZ}_{a_2c_2}\mathrm{CZ}_{a_3b_3}$ in the second layer.
    The resulting Pauli expansion contains
    $Z_{a_1}Z_{a_2}Z_{a_3}X_{b_1}X_{c_1}$, leaving a logical $Z$
    on block $a$ after correction of the $X$ errors.
    Red annotations indicate faults, propagated CZ factors,
    and representative $Z$ components of their Pauli expansions.
    }
    \label{fig:ccz-fault-propagation}
\end{figure}

The fault-tolerance limitations of the CCZ circuits for tricycle codes with depth larger than one were identified in concurrent
work~\cite{menon2026factories}. In this section, we analyze the depth-two circuit for the $[[27,3,3]]$ code explicitly and explain how extra Z syndromes help with error detection.

We consider a CCZ circuit on three $[[27,3,3]]$ CSS code blocks:
two layers of 27 disjoint CCZ gates, each coupling one qubit from each block.
In our noise model, after each gate, a non-identity three-qubit Pauli occurs with certain probability. Thus, weight-one, -two, and -three gate
errors are all first-order faults. We enumerate each gate fault individually,
and additionally test $X,Y,Z$ on each of the 81 input qubits to represent
residual preparation errors, without assigning these input faults an
occurrence rate. This gives $54\times63+81\times3=3645$ fault-location/Pauli
pairs. All subsequent syndrome measurements are ideal.

We insert each error at its specified location and propagate it through all
later CCZ gates. An $X$ propagated through CCZ acquires a CZ on the other
two qubits. We express each resulting CZ as a sum of Pauli operators on
those two qubits, $\mathrm{CZ}=(II+ZI+IZ-ZZ)/2$. Postselection on
trivial $X$ syndrome keeps only terms that commute with every $X$ stabilizer
in all three blocks. In this enumeration, at most one Pauli term per fault
survives. The probability of passing the $X$ checks, conditional on that
fault occurring, therefore equals the squared magnitude of the surviving
term's coefficient. We denote this probability by $P_X$.

There are 513 faults with non-zero probability of passing X checks, i.e., $P_X>0$, classified below and also see examples in \cref{fig:ccz-fault-propagation}. The other 3132,
including every nonidentity purely $Z$ fault, are rejected with certainty.
Gate faults are inserted immediately after the indicated layer's gate.
The three letters specify the Pauli operators in block order.
\begin{center}
\small
\begin{tabular}{llrr}
\toprule
Location & Pauli error & Count & $P_X$ \\
\midrule
Input & Single-qubit $X$ & 81 & $1/16$ \\
Layer 1 & $XII,IXI,IIX$ & 81 & $1/4$ \\
Layer 1 & $XXI,XIX,IXX$ & 81 & $1/16$ \\
Layer 1 & $XXX$ & 27 & $1/64$ \\
Layer 1 & $XZX,ZXX,XXZ$ (selected locations) & 27 & $1/16$ \\
Layer 1 & $XYX,YXX,XXY$ (selected locations) & 27 & $1/64$ \\
Layer 2 & $XII,IXI,IIX,XXI,XIX,IXX,XXX$ & 189 & $1$ \\
\bottomrule
\end{tabular}
\end{center}
For each row containing $Z$ or $Y$, exactly one of the three listed errors
per gate can pass (nine gate locations per listed error). Only the 189 last-layer faults pass with certainty.

When they pass the $X$ checks, the 54 first-layer faults with errors
$XZX$, $ZXX$, $XXZ$, $XYX$, $YXX$, or $XXY$ leave a nontrivial logical $Z$ after
their $X$ component is corrected. Their surviving $Z$ component commutes
with every $X$ stabilizer but is not a product of $Z$ stabilizers.
The remaining 459 surviving terms contain only $X$ errors and are removed
by correcting their $X$ component. All 513 surviving terms contain an
$X$ component detected by at least one $Z$ stabilizer.
Consequently, adding $Z$ syndrome postselection rejects
\emph{all 513} surviving fault branches, including all 54 with a logical
$Z$ component.

\paragraph{Recovering distance 3 with Z syndrome.}
We additionally show that by adding additional Z syndrome, no pairs of errors can be undetected while leaving a logical error. Consider the 3645 Pauli errors from the input of the CCZ circuit and the errors associated with the CCZ gates. Note that the $[[27,3,3]]$ code has X distance $d_X=6$ and X stabilizer check weight $6$. Therefore, the Z syndrome can detect X errors up to weight 5. An individual input error or CCZ error can at most insert weight-1 $X$ errors on each code block. As a result, the only undetectable pairs of errors are either with pure $I$ and $Z$ components, or with X components haveing completely overlapping support. This substantially reduces the search space. For any such pairs of errors, we have the propagated error $D$ at the end that is a product of Pauli $Z$ and $CZ$ errors. One could expand $D$ in the Pauli operator basis and group them according to stabilizer equivalence, and eventually have
\begin{equation}
    \Pi D\Pi=\sum_{L_Z} \alpha_{L_Z} L_Z\Pi
\end{equation}
where $\Pi$ is the code space projector, $L_Z$ is some logical Z operator including the identity, up to stabilizer equivalence. The condition for detectability of such errors is then $\alpha_{L_Z}=0$ for all $L_Z\notin S$. Enumerating all such pairs of errors confirmed that the detectability condition is satisfied and therefore the circuit has distance 3 under the final perfect X and Z syndrome. This result applies to protocol D (see \cref{sec:syndrome-protocols}), where fault-tolerant Steane-style syndrome extraction is used.

We also checked the detectability condition for the protocol E (see \cref{sec:syndrome-protocols}) with the added CNOT gate errors. There, the detectable condition is that the residual error must be correctable if it is undetected by the X syndrome from the final transversal X measurement and the Z syndrome from the shared Steane code block. We found pairs of CCZ errors, each containing weight-2 X components, and together they can lead to undetectable logical errors. However, these pairs of errors only correspond to a negligible fraction of the logical infidelity according to our biased CCZ error model, and it performs the best among the A, C, D, E protocols (B is the ideal noiseless syndrome extraction case).

\subsection{Z syndrome-extraction protocols}
\label{sec:syndrome-protocols}

In this section, we present some Z syndrome extraction gadgets after the constant-depth CCZ circuit that we explored in this work, motivated by the circuit-distance recovery via incorporating Z syndrome mentioned in \cref{sec:ccz-single-faults}. We also note that in a concurrent work \cite{menon2026factories}, one could alternatively recover circuit distance using code concatenation.

We first briefly review that we initially prepare three $[[27,3,3]]$ code blocks with every logical
qubit in $|+\rangle$. For each block, we initialize the 27 physical
qubits in $|+\rangle$ and measure the $X$ and $Z$ stabilizers using
a syndrome-extraction circuit with six CNOT layers. We retain only
preparations with a trivial measured $X$ syndrome. The 27 measured
$Z$ checks satisfy 11 independent parity constraints, called
\emph{metachecks}, which identify inconsistent syndrome records.
We consider two preparation variants: repair the measured $Z$
syndrome using these constraints, or discard preparations that
violate any metacheck. In either variant, we decode the resulting
syndrome and apply the inferred $X$ correction to the data.
This prepares the logical plus state up to residual errors, which
are retained and fed into the next stage of the simulation. The three accepted blocks then undergo the depth-two CCZ circuit. See \cref{sec:simulation-details} for details.

All methods postselect on trivial final data $X$ syndrome.
For fidelity evaluation, the final data $Z$ syndrome is measured
ideally and decoded to correct remaining $X$ errors. The methods
differ in the additional $Z$-syndrome postselection after CCZ:

\noindent\textbf{A: No additional extraction.}
There is no additional $Z$-syndrome measurement circuit or
$Z$-syndrome postselection. The final $Z$ syndrome is used
only for correction.

\noindent\textbf{B: Ideal extraction.}
We additionally require trivial ideal data $Z$ syndrome in every
block. This provides a reference with perfect syndrome extraction.

\noindent\textbf{C: Bare-ancilla extraction.}
Each of the 27 weight-four $Z$ checks per block is measured using
a separate $|0\rangle$ ancilla, with data qubits controlling CNOTs
onto that ancilla. The CNOT ordering at this stage is obtained by deleting the CNOTs corresponding to the X checks from the full syndrome extraction during logical $\ket{\overline{+}}$ state preparation. We measure the ancillas in the $Z$ basis,
repair the extracted syndrome using metachecks, and retain the
shot only if the repaired syndrome is zero in every block.
This repair step is used for both preparation variants.

\noindent\textbf{D: Steane extraction.}
Each data block is coupled to its own encoded plus-state ancilla
block, prepared by the same procedure as the data. A transversal
CNOT applies one CNOT from each data qubit to the corresponding
ancilla qubit. All three block pairs interact in parallel.
Ancilla qubits are then measured in the $Z$ basis, and we require a 
trivial inferred $Z$ syndrome for each block.

\noindent\textbf{E: Shared Steane extraction.}
One encoded plus-state ancilla block is coupled successively
to the three data blocks by transversal CNOTs and measured once
at the end. We require a trivial syndrome inferred from this
ancilla readout. With ideal ancillas and couplings, this tests
the bitwise sum modulo two of the three data-block $Z$ syndromes.
In comparison, D tests each syndrome separately. Thus E uses fewer prepared
ancilla blocks, but matching syndromes in different blocks
can cancel.

\begin{center}
\begin{tabular}{llrr}
\toprule
Method & Final $Z$ extraction & Added CNOTs & CNOT layers \\
\midrule
A & None                            & 0       & 0       \\
B & Ideal reference                 & $0^{*}$ & $0^{*}$ \\
C & Bare ancillas                   & 324     & 6       \\
D & Three Steane ancilla blocks     & 81      & 1       \\
E & One shared Steane ancilla block & 81      & 3       \\
\bottomrule
\end{tabular}
\end{center}
These counts cover all three data blocks and only the final
extraction couplings, excluding state preparation and readout.
In C, the count is $3\times27\times4=324$. In D and E, it is
$3\times27=81$. Each of C--E also requires one parallel
ancilla-measurement stage. For B, $0^{*}$ denotes the absence
of an explicit simulated extraction circuit.

\putbib[references]
\end{bibunit}
\end{document}